\documentclass[11pt]{article}

\usepackage{geometry}
\usepackage{amssymb}
\usepackage{mathrsfs}
\usepackage{amsmath}
\usepackage{float}
\usepackage{color}
\usepackage{algorithm}
\usepackage[noend]{algpseudocode}
\usepackage{hyperref}
\usepackage[normalem]{ulem}

\usepackage{times}

\usepackage[numbers]{natbib}

\newcommand{\sd}[1]{{\color{blue} Shahar: #1}}

\newcommand{\mbc}[1]{[{\color{red} MB: #1}]}

\newcommand{\rk}[1]{{\color{orange} [Ron: #1]}}

\newcommand{\cout}[1]{}

\newtheorem{theorem}{Theorem}[section]

\newtheorem{lemma}[theorem]{Lemma}
\newtheorem{proposition}[theorem]{Proposition}

\newtheorem{claim}[theorem]{Claim}

\def\squarebox#1{\hbox to #1{\hfill\vbox to #1{\vfill}}}
\newcommand{\qed}{\hspace*{\fill}
	\vbox{\hrule\hbox{\vrule\squarebox{.667em}\vrule}\hrule}\smallskip}

\newenvironment{proof}{\noindent{\bf Proof:~~}}{\(\qed\)}
\usepackage{graphicx}
\newcommand{\eps}{\varepsilon}

\newcommand{\comment}[1]{}
\newcommand{\stam}[1]{}

\newcommand{\REV}{{\textsc{Rev}}}

\newcommand{\D}{\mathcal{D}}
\newcommand{\Q}{\mathcal{Q}}

\newcommand{\Fd}{\mathcal{F}_d^p}
\newcommand{\Fp}{\mathcal{F}^p}
\newcommand{\Fk}{\mathcal{F}_k^p}
\newcommand{\Fkp}{\mathcal{F}_{k+1}^p}
\newcommand{\Fkm}{\mathcal{F}_{k-1}^p}
\newcommand{\Qv}{{\Q_v}}

\newcommand{\DQ}{\mathcal{D}_{|\Q}}
\newcommand{\DQv}{\mathcal{D}_{\Qv}}

\newcommand{\Qvk}{\mathcal{K}_{\Qv}}

\newcommand{\R}{\mathbb{R}}

\newcommand{\LQ}{L_\mathcal{Q}}
\newcommand{\LQp}{L_\mathcal{Q'}}
\newcommand{\AQ}{A_\mathcal{Q}}
\newcommand{\AQp}{A_\mathcal{Q'}}
\newcommand{\eQ}{\eps_\mathcal{Q}}
\newcommand{\eQp}{\eps_\mathcal{Q'}}

\newcommand{\Bv}{\mathcal{B}_{v}}
\newcommand{\Bvp}{\mathcal{B}_{v'}}
\newcommand{\Rv}{\mathcal{R}_{v}}
\newcommand{\Rvp}{\mathcal{R}_{v'}}
\newcommand{\Gv}{G_{v}}
\newcommand{\Gvp}{G_{v'}}

\newcommand{\Byes}{\mathcal{B}_{\Qv}^y}
\newcommand{\Bno}{\mathcal{B}_{\Qv}^n}
\newcommand{\E}[2][{}]{\mathbb{E}_{{#1}}\left[ #2\right]}

\newcommand{\Sdp}{S_d^p}
\newcommand{\Spk}{S_k^p}
\newcommand{\Spkp}{S_{k+1}^p}
\newcommand{\Spkm}{S_{k-1}^p}
\newcommand{\Upk}{U_k^p}

\renewcommand{\O}[1]{\mathcal{O}\left( #1\right) }

\newcommand{\Om}[1]{\Omega\left( #1\right) }

\newcommand{\items}{M}
\newcommand{\menu}{{\cal M}}
\newcommand{\opt}[1]{$#1$-optimal}

\newcommand{\set}[1]{\left\{#1\right\}}

\usepackage{tikz}
\usetikzlibrary{arrows.meta,chains,decorations.pathreplacing}

\begin{document}

	\title{Simplicity in Auctions Revisited: The Primitive Complexity \footnote{Work done while all authors were at Microsoft Research. The second author was also partially supported by BSF grant 2016192 and ISF grant 2185/19.}
}
    
	\author{
	Moshe Babaioff \thanks{Microsoft Research. Email: moshe@microsoft.com.} \and Shahar Dobzinski\thanks{Weizmann Institute of Science and Microsoft Research. Email: shahar.dobzinski@weizmann.ac.il.} \and Ron Kupfer\thanks{Harvard University. Email: ron.kupfer@mail.huji.ac.il.}
	}
	\maketitle

	\begin{abstract}

		In this paper we revisit the notion of simplicity in mechanisms. We consider a seller of 
		$m$ heterogeneous items, facing a single buyer with valuation $v$. We observe that previous attempts to define complexity measures often fail to classify mechanisms that are intuitively considered simple (e.g., the ``selling separately'' mechanism) as such. We suggest to view a menu as simple if a bundle that maximizes the buyer's profit 
		can be found by conducting a few primitive operations that are considered simple.  
		The \emph{primitive complexity of a menu} is the number of primitive operations needed to (adaptively) find a profit-maximizing entry in the menu.  In this paper, the primitive operation that we study is essentially computing the outcome of the ``selling separately'' mechanism. 
		
Does the primitive complexity  capture the simplicity of other auctions that are intuitively simple? We consider \emph{bundle-size pricing},  a common pricing method in which the price of a bundle depends only on its size.
		Our main technical contribution is determining the primitive complexity of bundle-size pricing menus in various settings. First, we connect the notion of primitive complexity  to the vast literature on query complexity. We then show that for any distribution $\D$ over weighted matroid rank valuations, even distributions with arbitrary correlation among their values, there is always a bundle-size pricing menu with low primitive complexity  that achieves almost the same revenue as the optimal bundle-size pricing menu. As part of this proof we provide a randomized algorithm that for any weighted matroid rank valuation $v$ and integer $k$, finds the most valuable set of size $k$ with only a poly-logarithmic number of demand and value queries. We show that this result is essentially tight in several aspects. For example, if the valuation $v$ is submodular, then finding the most valuable set of size $k$ requires exponentially many queries (this solves an open question of Badanidiyuru et al. [EC'12]). We also show that any deterministic algorithm that finds the most valuable set of size $k$ requires $\Omega(\sqrt m)$ demand and value queries, even for additive valuations.
	\end{abstract}
	\thispagestyle{empty}
\newpage
    \setcounter{page}{1}

	\section{Introduction}\label{sec:intro}

\cout{
\begin{itemize}
    \item We need to discuss the relation of the paper to "The Randomized Communication Complexity of Randomized Auctions" by Aviad Rubinstein and Junyao Zhao

    \item "The paper does not provide a lower bound in the setting (weighted matroid-rank) where the polylog positive result is given. The lower bound results are for restricted mechanisms or for more general valuations." -can we give a polylog lower bound ?
    
    \item should we follow reviewer 3 and expand the discussion of prior complexity measures and their weaknesses? 
    \item we are allowed 18 pages. Do we want to move stuff back from the appendix to the body?
    
\end{itemize}
}

The search for simple mechanisms is a central theme in the Mechanism Design literature (e.g., \cite{Ronen01,hartline2009simple,dutting2011simplicity,cai2017simple,rubinstein2018simple})\cout{HartlineR09}. While complex mechanisms might be theoretically superior in terms of, e.g., extracting more revenue from the participants, they are often inferior in practice. For example, complicated rules might be harder for the designer to implement and for the bidders to understand and to interact with, thus making them less attractive. 

Of course, simplicity is a vague and elusive notion, and, unfortunately, 
there is little hope of finding a formal mathematical definition 
that sharply separates simple mechanisms from complex ones. 
Furthermore, a mechanism that is simple in one setting or for one group of participants might not be considered simple in other situations. To a large extent, simplicity is in the eye of the beholder. But to an even larger extent, the ``I know it when I see it'' test perfectly applies here.

Yet, a mathematical treatment of simplicity in auctions must be based on some formal definition, as imperfect as may be. Broadly speaking, many papers identify simplicity with particular forms of mechanisms (e.g., second price auctions \cite{moldovanu1998goethe,thompson2013revenue}, ascending auctions \cite{cramton1998ascending,gul2000english,mishra2007ascending}, posted prices auctions \cite{chawla2010multi,hartline2009simple})\cout{HartlineR09} and analyze these classes of mechanisms. 

In some settings, e.g., selling a single item, identifying simplicity with a specific
auction format is an extremely useful idea. The canonical example is Myerson's optimal auction characterization that shows that second price auctions with reserve are optimal when values are drawn i.i.d. from a regular distribution; Other papers show that such auctions are approximately optimal in some other settings \cite{hartline2009simple,alaei2014bayesian,alaei2019optimal,chawla2010multi}. Yet, the more complex the setting is, the less likely it is that a rigid list of permissible ``simple'' auction formats will provide optimal or approximately optimal results. 
Note also that a binary classification of auctions as either ''simple'' or not,
does not allow the ranking and quantification of different auctions: which auction is simpler, a ``selling separately'' auction where each item has a different price, or the auction that sells any bundle of $10$ items at price $1$?  And by how much?

This calls for using simplicity \emph{measures} as an additional tool for analyzing auctions -- focusing on quantitative approaches to simplicity. One of the most influential measures, the \emph{menu complexity}, was suggested by Hart and Nisan [\citeyear{hart2019selling}]. Consider a pricing problem in which a monopolist wants to sell a set $\items$ of $m$ heterogeneous items. A {deterministic} \emph{menu} $\menu$ is a set of pairs $(S,p_S)$, in which every such pair specifies the price of the bundle 
$S\subseteq M$. The buyer has a valuation $v:2^\items\rightarrow \mathbb R_{\geq 0}$ that specifies her value for every possible subset of the items. The valuation is drawn from some known distribution $\mathcal D$. Given deterministic menu $\mathcal M$, the buyer is assigned a bundle $O$ that maximizes her profit, that is, bundle $O\in \arg\max_{S\subseteq \items} v(S)-p_{S}$. The revenue of the mechanism is the expected payment of the buyer. The \emph{menu complexity} of the mechanism is the minimal number of pairs (of a bundle and its price) needed to describe the mechanism.


In many cases the notion of menu complexity captures 
the simplicity of auctions very well. Menus with few entries tend to be ``simpler'' than menus with many entries, whatever the precise meaning of simplicity is. Indeed, in recent years we have seen the notion of menu complexity grows in popularity and being extended to more settings \cite{babaioff2021menu,chawla2020menu,dughmi2014sampling,gonczarowski2018bounding,saxena2018menu} and to richer classes of valuations \cite{eden2021simple,rubinstein2018simple}. In general, the literature tends to draw the simplicity/complexity borderline by treating menu complexity $poly(m)$ as a proxy for simplicity. 

However, intuitive simplicity and menu complexity do not always go hand by hand. Consider an additive valuation over items and the mechanism that sells item separately, with item $j$ sold at price $p_j$. The menu complexity of this mechanism is exponential\footnote{The complexity of the ``selling separately'' menu is exponential also for the symmetric menu complexity \cite{kothari2019approximation}, which is a generalization of the menu complexity.}, as for each set $S$ out of the $2^m-1$ non-empty sets, it needs to list it with price $\sum_{j\in S} p_j$. Yet, this menu is intuitively very simple. 
Moreover, a simple variant of this mechanism was shown to have very attractive properties. 
In fact, \citet{babaioff2020simple} show that for additive valuations, when the value of every item $j$ is drawn independently from a known distribution $\mathcal D_j$, then one of the following mechanisms extracts a constant fraction of the optimal revenue: 
sell the bundle of all items at the monopolist price (with respect to the distribution of the bundle of all items), or separately sell each item $j$ at the monopolist price of the distribution $\mathcal D_j$. Most would agree that the Babaioff et al. mechanism is simple, but unfortunately its menu complexity is huge due to the ``selling separately'' component.  

Another example of the limits of the notion of menu complexity can be found in the popular and practically-used ``bundle-size'' pricing (see, e.g., \cite{chu2011bundle,abdallah2021large} and their followups) which prices all bundles of the same size at the same price. 
Although this menu is intuitively simple, its menu complexity is high. 
Indeed, observe that 
when the price of every bundle of size $\frac m 2$ is $1$, 
the menu complexity of this mechanism is exponential, since each of the exponentially many bundles of size $\frac m 2$ requires an entry in the menu.





\paragraph{The Primitive Complexity of Auctions.} 
We have exhibited several examples of mechanisms that pass the ``I know it when I see it'' test for simplicity, yet have high menu complexity. This calls for a more nuanced approach toward measuring simplicity. Before presenting our approach, we would like to stress again that an ``ultimate'' mathematical definition of simplicity is unlikely to exist. In all likeliness, inevitably, as any other simplicity notion, our new notion will fail for some mechanisms that ``should'' be considered simple and will include mechanisms that ``should'' be classified as complex. Yet, we believe that our approach would better capture the simplicity of many mechanisms. For other mechanisms, different approaches, possibly tailored to the specific application, might be useful. 

The basic intuition that leads our work is that mechanisms are often considered simple if they can be implemented by applying only a small number of primitive operations that are considered ``simple''. In our case, the primitive operation is computing the outcome of the ``selling separately'' auction, perhaps the canonical example for a simple auction that is not captured as such by the notion of menu complexity. That is, given prices per item $p_1,\ldots, p_m$, return a bundle $S$ that maximizes the buyer's profit (return $S\in\arg\max_T v(T) - \sum_{j\in T}p_j$), as well as the value $v(S)$ of the bundle $S$. The \emph{primitive complexity of a menu} is the number of times the primitive operation has to be {(adaptively)} applied to find 
a bundle that maximizes the buyer's profit for any given valuation. 
The fewer times the primitive operation has to be invoked, the simpler the menu is. 

The primitive complexity was defined here for deterministic algorithms using worst case approach on valuations, yet the definition naturally extends to randomized algorithms and to valuations sampled from a Bayesian prior.\footnote{A related notion is the \emph{randomized} communication complexity of finding the profit-maximizing bundle in the menu. See \cite{rubinstein2021randomized}.} Similarly, the definition can be extended by restricting the valuations to belong to a specific class (e.g., only additive or submodular valuations). Note that in principle, only the number of queries that the algorithm makes is restricted, not the running time, though all algorithms that we develop in this paper are computationally efficient. The definition of primitive complexity can also be naturally extended to randomized menus, i.e., menus that allow lotteries over bundles.

Our focus is in understanding whether the primitive complexity of the intuitively-simple class of bundle-size pricing menus is indeed low (mostly when the primitive operation is computing the outcome of the ``selling separately'' mechanism). Towards this end, we rely on (and advance) the literature on query complexity and valuation functions. 
In fact, since we will observe that the primitive complexity is essentially equivalent to a  query complexity of finding a buyer's profit-maximizing bundle (as discussed below), our work also suggests that simple mechanisms are those for which a profit-maximizing bundle can be ``easily'' found.




\paragraph{Connection to Query Complexity.}  We now discuss the connection of primitive complexity to query complexity. Recall that as usual in algorithmic game theory, the size of a naive description of the valuation $v$ is exponential in the number of items. Thus it is common to assume that $v$ is given as a black box that can only answer a limited number of types of queries. The two standard queries are  \emph{value queries} (given $S$, what is $v(S)$?) and \emph{demand queries} (given item prices $p_1,\ldots, p_m$, return a bundle $S$ that maximizes the profit of the buyer. That is, find $S\in\arg\max_T v(T) - \sum_{j\in T}p_j$). It is not hard to see that the outcome of a ``selling seperately'' operation can be simulated by a demand query followed by a value query. Also note that a value query to a bundle $S$ can be simulated by considering the outcome of one ``selling seperately'' operation 
that assigns a price $0$ for every item in $S$ and $\infty$ for any other item (as the operation returns the value of the demanded set). 
Thus, the primitive complexity and the query complexity are related up to a constant multiplicative factor.

Value queries are extensively used in various optimization problems  
\cite{nemhauser1978analysis,calinescu2011maximizing,vondrak2008optimal}. Demand queries are standard in the algorithmic game theory literature and appear naturally in various posted prices auctions \cite{assadi2019improved,feldman2014combinatorial}, as the separation oracle needed to solve the natural LP relaxation for combinatorial auctions \cite{nisan2006communication}, and in various (not necessarily incentive compatible) approximation algorithms \cite{dobzinski2006truthful,feige2006approximation,feige2010submodular}.\footnote{Mathematically speaking, if $b$ is the known number of bits used to represent numbers then a value query can be computed with $b$ demand queries, whereas computing a demand query might require $exp(m)$ value queries \cite{blumrosen2010computational}. To some extent, some would argue that in practice it is common to solve a demand query (what would you buy in the grocery store?) where as value queries are harder (what is your value for $12$ eggs, bread, and a bottle of orange juice?)}

With this interpretation of primitive complexity in mind, the primitive complexity of a {deterministic} menu is at most its menu complexity: every menu with menu complexity $c$ can be implemented by making $c$ value queries to query $v(S)$ for each bundle $S$ which has an entry in the menu. The converse is far from being true: 
the primitive complexity of the mechanism that separately sells each item $j$ at price $p_j$ is just $1$, while its 
menu complexity is exponential.

\paragraph{Our Results.} In this paper we analyze the primitive complexity of the extensively studied class of bundle-size pricing menus. Recall that bundle-size pricing menu gives a price of $p_r$ for every number of items $r\in[m]$. We start by considering the family of additive valuations ($v(S)=\sum_{j\in S}v(\{j\})$ for every bundle $S$). For additive valuations, the primitive complexity of every menu is at most $m$: 
querying the value $v(\{j\})$ of every item $j$ gives the entire valuation and thus suffices to compute a profit-maximizing bundle. 
Thus, in the context of additive valuations, simplicity will be captured by sub-linear primitive complexity, ideally achieving complexity that is  poly-logarithmic in $m$, or even a constant. 
We prove that the primitive complexity of (approximately) maximizing the revenue is 
much better than linear in $m$, showing that it is only poly-logarithmic. 

\vspace{0.08in} \noindent \textbf{Theorem I: } Let $\D$ be some distribution over additive 
valuations. Then, for any $\eps>0$ there is a bundle-size pricing menu $\mathcal M$ with primitive complexity $poly(\log m, \frac 1 \eps)$ such that the revenue of $\mathcal M$ is in expectation at least $(1-\eps)$ of the revenue of any other bundle-size pricing menu on $\mathcal D$.

\vspace{0.08in} \noindent Note that the distribution $\D$ can be arbitrary. In particular, we do not assume that the values of the items are drawn from independent distirbutions\footnote{
For additive valuations, when item values are sampled independently, \citet{babaioff2021menu} prove that 
the auction that separately sells each item can be well approximated by a menu with a polynomial size. In contrast, our result does not assume independence, it holds for any menu, and moreover, as we will see later, generalizes to much richer classes of valuations.
}.
To prove the theorem we 
show that for any bundle-size pricing menu $\mathcal M'$, there is a low complexity bundle-size pricing menu $\mathcal M$ with revenue comparable to that of $\mathcal M'$. 
Specifically, we show that although $\mathcal M$ might contain many distinct prices and bundles sizes, 
only $poly(\log m,\frac  1 \eps)$ bundle sizes need to be considered (in expectation over the distribution $\mathcal D$) to find a profit-maximizing bundle.
Thus, the profit-maximization problem was reduced to the problem of finding a set of highest value for a given set size (maximization under a cardinality constraint), using ``selling seperately'' operations or, almost equivalently, value and demand queries: given a valuation $v$ and bundle size $k$, find a bundle $S_k$ that maximizes $v(S)$ subject to $|S_k|=k$ (in this case we say that $S_k$ is a \emph{\opt{k}} set). Let $p_k$ be the price of bundles of size $k$ in $\mathcal M$. A profit maximizing bundle of $\mathcal M$ is in $\arg\max \{v(S_k)-p_k\}$ (or the empty set, if this maximum profit is negative). 
Thus, to complete the proof we prove the following algorithmic result: 

\vspace{0.08in} \noindent \textbf{Theorem II: } Fix some bundle size $k$. There exists a randomized algorithm that given an additive valuation $v$ finds a bundle $S_k\in \arg\max_{S:|S|=k}v(S)$ by making, in expectation,  $poly\log(m)$ value and demand queries (the expectation is over the randomness of the algorithm).

\vspace{0.08in} \noindent Note that our algorithm for finding a \opt{k} set does not assume that the valuations are drawn from some distribution: the guarantee is in the worst case, for every possible additive valuation. Furthermore, our algorithm \emph{always} finds a value-maximizing bundle of size $k$, randomization is only use to accelerate the running time. Moreover, a significant challenge in developing our algorithms is that we want the algorithms to work with every implementation of the demand oracle. That is, the demand query is required to return a profit-maximizing bundle, but if there are several such bundles we want our algorithms to work with any implementation of the tie-breaking rule, even an adversarial one. 

We show that our algorithm for additive valuations is qualitatively optimal in multiple respects: 
\begin{itemize}
\item \textbf{Deterministic algorithms with value queries:} We show that any deterministic algorithm that uses only value queries must make at least $m-1$ queries in order to find the item with the highest value ($k=1$). 

\item \textbf{Randomized algorithms with value queries:} We provide two different proofs that show with no demand queries, $\Omega (\frac m {\log m})$ value queries are needed to find the $k$-optimal bundle, even if the algorithm is randomized. One proof assumes $k=\frac m 2$ and is based on a simple counting argument. The second proof is based on a slightly more involved communication complexity argument is and shows this impossibility even for the simple case of $k=1$ (i.e., finding the most valuable item).

\item \textbf{Randomization is required: } We prove that even if demand queries are allowed but the algorithm must be deterministic, then $\Omega(\sqrt m)$ queries are required.


\end{itemize}
The proofs of the first two impossibilities are easier than the proof of the last result, which is more subtle and involved. 
We then move on to consider richer valuation classes, starting with weighted matroid rank functions:

\vspace{0.08in} \noindent \textbf{Theorem III: } Fix a bundle size $k$. There is a randomized algorithm that finds a bundle $S_k\in \arg\max_{S:|S|=k}v(S)$ for every weighted matroid rank valuation $v$ and makes in expectation $poly\log(m)$ value and demand queries.

\vspace{0.08in} \noindent Using Theorem III we are able to extend Theorem I to hold for weighted matroid rank valuations, not just additive ones. For weighted matroid rank valuations, the greedy algorithm finds a \opt{k} set with $\Theta(m\cdot k)$ value queries. We thus see that with demand queries it is possible to find a \opt{k} set exponentially faster.
Furthermore, even if all weights of the items are either $0$ and $1$, the number of matroid rank functions is doubly exponential \cite{knuth1974asymptotic}. Nevertheless, we find a $k$-optimal set in only $poly\log(m)$ queries. 

This paper is not the first to consider maximization subject to cardinality constraint with value and demand queries. The first was \cite{badanidiyuru2012optimization}, and it considered richer classes: submodular, XOS, and subadditive valuations. However, while the current paper considers exact optimization algorithms, the focus of \cite{badanidiyuru2012optimization} was in approximation algorithms. Approximation algorithms are not useful for finding a profit maximizing bundle in a menu, which must be done exactly, otherwise the incentive constraints are likely to be violated. One of the main results of \cite{badanidiyuru2012optimization} is a $\frac 9 8$-approximation algorithm for maximizing a submodular function subject to a cardinality constraint using value and demand queries (recall that with value queries only, the greedy algorithm provides an approximation ratio of $\frac e {e-1}$ \cite{nemhauser1978analysis}). However, the paper \cite{badanidiyuru2012optimization} proves no impossibility at all for this setting, asking whether an \emph{exact} solution can be found with polynomially many value and demand queries. We solve this open question:

\vspace{0.08in} \noindent \textbf{Theorem IV: } Fix a randomized algorithm $A$ for maximizing a submodular function subject to cardinality constraint that succeeds with constant probability. Then, $A$ makes at least $exp(m)$ value and demand queries. 

\vspace{0.08in} \noindent This result highlights the importance of the valuations class in measuring the primitive complexity: the primitive complexity of the menu that corresponds to maximization subject to a cardinality constraint is poly logarithmic for the class of additive valuations, but exponential for the richer class of submodular valuations.


\paragraph{Connections to other Problems.} The problem of maximization subject to cardinality constraint with demand queries has interesting connections to some well-studied problems. One such problem is unordered partial sorting \cite{chambers1971algorithm}: we are given an array that contains $n$ numbers and the goal is to find a set with $k$ highest numbers (in any order). Note that unordered partial sorting is equivalent to finding a $k$-optimal set in an additive valuation. 

Unordered partial sorting can be solved in $\O{m}$ time by QuickSelect \cite{hoare1961find}, a variant 
of QuickSort. Note that the pivot procedure of QuickSort and QuickSelect -- divide the array into two, one includes all numbers bigger than some $p$, and the other contains all number smaller than $p$ -- is essentially a demand query at price $p$ per item (with some tie-breaking rule that depends on the implementation). Indeed, QuickSelect makes in expectation $\O{\log m}$ pivot calls. Our lower bound for deterministic algorithms implies that randomization is essential to QuickSelect and other pivot-based algorithms in the sense that any deterministic algorithm requires $\Omega(\sqrt{m})$ pivot and value queries. 
Note that algorithms like IntroSelect \cite{musser1997introspective} use more advanced methods to deterministically ensure a good selection of the pivot for QuickSelect. However, these algorithms use many value queries, and our results show that this is unavoidable.

There are also connections to various coin weighting problems. For example, consider the following problem studied in \cite{bshouty2009optimal,djackov1975search,lindstrom1975determining,du2000combinatorial}: we are given $n$ coins. We know that $d$ of them are counterfeit. The weight of each real coin is $w_1$ and the weight of each counterfeit coin is $w_2$. We are also given a spring weight. How many weightings are needed to find all counterfeit coins? Note that every use of the spring weight is equivalent to a value query. If $d$ is big then a simple counting argument -- similar to the one that prove that even randomized algorithms must make many value queries -- shows that many weightings are needed. However, our communication-complexity based proof shows that if there are three types of coins with $w_3>w_2>w_1$ then finding just one coin with weight $w_3$ requires almost linear number of weightings, a result that was not known before, to the best of our knowledge. 

\paragraph{Future Directions.} In this paper we introduced a new measure of complexity for auctions: the primitive complexity. We have examined this notion in the context of a pricing problem, and obtained some algorithms with low primitive complexity as well as some impossibilities. Obviously, studying the primitive complexity of other problems is an exciting future direction.

We propose a number of open questions. We have shown that bundle-size pricing menus have poly-logarithmic primitive complexity if the valuation belongs to the class of weighted matroid rank functions, and that if the valuation belongs to the class of submodular valuations the primitive complexity is exponential. The class of gross substitutes valuations contains all weighted matroid rank functions and is contained in the class of submodular valuations. We know that the greedy algorithm finds a $k$-optimal bundle with polynomially many value queries, but can we find a $k$-optimal bundle with poly-logarithmic number of value and demand queries?

In addition, it will be extremely interesting to understand whether bundle-size pricing can well approximate the revenue that can be obtained by any deterministic mechanism. Of course, for this question to make sense we have to consider some kind of symmetry in the distribution. 
For additive valuations, if the values of the items are sampled i.i.d. then the mechanism of \citet{babaioff2020simple} already implies that bundle-size pricing can provide a constant fraction of the optimal revenue. But what if the joint distribution of item values is symmetric, yet item values are not sampled i.i.d.? Can bundle-size pricing provide a constant fraction of the optimal revenue that can be achieved by a deterministic mechanism then? See also \cite{BNR18} for some related work.

Finally, a fascinating direction is to obtain mechanisms with polylogarithmic primitive complexity that obtain $(1-\epsilon)$ fraction of the optimal revenue. We do not know how to obtain such a mechanism even for independent distributions and additive valuations.







	\section{Model and Preliminaries}\label{sec:model}
\paragraph{Valuations.}
Given a set $\items = [m] = \set{1,2,\ldots, m}$ of $m$ indivisible {\em items}, 
a \emph{valuation function}  $v:2^\items\rightarrow\R_{+}$ 
determines a non-negative \emph{value} $v(S)$ for each bundle $S\subseteq \items$. 
We make the standard assumptions that any valuation function $v$ is normalized ($v(\emptyset)=0$) and monotone (for $S\subseteq T$ it holds that $v(S)\leq v(T)$).
With a slight abuse of notation, for valuation 
$v$ and an item $a\in \items$  we use $v(a)$  to denote $v(\set{a})$. 
We consider several standard classes of valuations (each of the classes is strictly contained in the class that follows it): 
\begin{itemize}
	\item A valuation $v$ is \emph{additive} if for all $S\subseteq \items$ we have that $v(S)=\sum_{a\in S}v(\set{a})$. In this case we may represent the function as a vector $v = (v_1,...,v_m)\in \R^{m}_{+}$,
where $v_i\geq 0$ is the value of the $i$'th item.
	\item 

A valuation function $v$ is called \emph{weighted matroid-rank valuation} if there exists a matroid\footnote{A \emph{matroid} is a pair $(M,\mathcal{I})$, with $\items$ being a finite set of elements and 
$\mathcal{I}\subseteq 2^{\items}$
is a non-empty family of subsets of $\items$ such that:
(1) If $B\subseteq A\subseteq  \items$ and $A\in \mathcal{I}$, then $B\in \mathcal{I}$ and
    (2) If $A,B\in \mathcal{I}$ and $|B|>|A|$, then there exists an item $b\in B\setminus A$ such that $A\cup\set{b}\in \mathcal{I}$. A set $A\in \mathcal{I}$ is called an {\em independent set}. 
 An independent set $A\in \mathcal{I}$ is called a {\em base} of the matroid, if it is not contained in any larger independent set.  
The \emph{rank} of the matroid is the size of any base (all have the same size).} 
over the set of elements $M$ and a weight function\footnote{A weight function $\omega:M\rightarrow \R_{+}$ assigns a weight to each element. 
The weight function is extended to sets as follows: the weight of a set $S\subseteq M$ is defined to be $\omega(S)=\max_{A\subseteq S,~A\in \mathcal{I}}\sum_{a\in A}\omega(a)$. 
An independent set $A\in \mathcal{I}$ is called a {\em maximal weight independent set} if there is no independent set of larger weight.
} $\omega:M\rightarrow \R_{+}$ such that  $v(S)=\omega(S)$ for every $S\subseteq \items$.

	\item A valuation $v$ is called \emph{submodular} if it exhibits the diminishing returns property, i.e., $v(S\cup \{a\})- v(S)\geq v(T\cup \{a\})- v(T)$ 
	for all    $S\subseteq T \subseteq \items$ and  $a\in \items$.
\end{itemize}

\paragraph{Maximization Subject to a Cardinality Constraint.}
Given a valuation $v$ over set $\items$, 
the \emph{cardinality maximization problem} with parameter $k$ is the problem of finding a maximum value set of size $k$. I.e., finding a set $S$ that is \emph{\opt{k}}: $S\in \arg\max_{S:|S|=k} v(S).$

\paragraph{Mechanisms.} 
%
%
We consider a setting with one seller holding a set $\items$ of $m$ items, 
that faces a single buyer with a valuation $v$.  
A {\em deterministic menu $\menu$} is a set of pairs $\set{S,p_S}$ of bundles and prices.\footnote{We assume that valuations are monotone non-decreasing (free disposal). 
Under this assumption, it is wlog to assume that 
for any $S,T\subseteq \items$ such that $S\subseteq T$ we have that $p_S\leq p_T$ 
(otherwise $S$ is never being sold and the menu entry can be removed).}
We assume that any menu includes the option of getting no item and paying $0$. 
Given a menu $\menu$, a set $S$ is a \emph{demanded set} (or a \emph{most profitable set})
of a buyer with valuation $v$ if  $S\in \arg\max_{S'\subseteq \items}v(S')-p_{S'}$.
The family of sets that are demanded are called the \emph{demand} of the buyer.
We assume that a buyer with valuation $v$ that faces $\menu$ selects a set that she demands, but make no assumption about how she picks between different demanded sets.
In a Bayesian setting, the valuation $v$ is drawn from 
a known distribution $\mathcal{F}$, and the revenue of the menu $\menu$  is measured in expectation over $\mathcal{F}$:
\begin{itemize}
	\item $\REV_{\menu}(v)$: the expected revenue of the seller from menu $\menu$ when the buyer's valuation is $v$, i.e., 
	if the buyer picks a demanded set $S$ with probability $g_S$ then the revenue is $\sum_{S\subseteq \items} p_S \cdot g_S$. 
	\item $\REV(\menu,\mathcal{F})$: the expected revenue where the expectation is over the buyer's valuation $v\sim\mathcal{F}$, that is $\REV(\menu,\mathcal{F})=\E[v\sim \mathcal{F}]{\REV_{\menu}(v)}$.
\end{itemize}

A specific class of mechanisms that is considered in this paper is 
\emph{bundle-size pricing}: the price of each bundle of size $r$ is $p_r$.
Since the valuations are monotone, we assume that for any two bundle sizes $q_i>q_j$ we have that $p_i>p_j$ (all inequalities are strict).



\paragraph{Queries.} 
In this paper we consider ``selling seperately'' operations: given a price $p_i$ for each item $i$, find some bundle in the demand, a bundle  in  $\arg\max_{S\subseteq \items}{\left(v(S)-\sum_{i\in S}p_i)\right)}$, 
and return this bundle and its value.

The literature on multi-item auctions has extensively studied two types of queries as means of accessing a valuation $v$ (which might have a large representation): value queries and demand queries.
A {\em value query} is given a set $S$ and simply returns $v(S)$, the value of the bundle $S$.
A {\em demand query} asks for a bundle of maximum profit at some given item prices, i.e., 
a most profitable set for the given prices (a demanded set).
Formally, the query is given an item-price vector $p = (p_1,\cdots,p_m)\in \R^{m}_{+}$ and  returns an arbitrary set $D \in \arg\max_{S\subseteq \items}{\left(v(S)-p(S)\right)}$ in the demand, where $p(S)=\sum_{i\in S}p_i$. We assume that the demand query also returns the value of the set $S$ (this can always be done at the cost of an additional value query).
We make no assumption about the way ties are broken between sets in the demand. Ties might be broken adversarially and this issue creates significant challenges which we need to address. 
When all the coordinates of $p$ have the same value $t\in\R$, we refer to the corresponding query as a {\em uniform demand query} for price $t$.

Obviously, a ``selling seperately'' operation can simulate a demand query. It can also simulate a value query for a bundle $S$: set the price of each item in $S$ to 0, and the price of every other item to $\infty$. It is also straightforward to see that any ``selling seperately'' operation can be simulated by one demand query followed by one value query. Hence, any algorithm that uses only $t$ ``selling seperately'' operations can be simulated with $2t$ value and demand queries, and every algorithm that uses $t$ value and demand queries can be implemented with $t$ ``selling seperately'' operations. Thus, we will freely switch between these two similar points of view.

\paragraph{The Primitive Complexity.}

In this paper we suggest to consider the \emph{primitive complexity} of menus.
Fix a class of valuations $\mathcal V$. The primitive complexity of an algorithm with respect to $\mathcal V$ is the maximal number of ``selling seperately'' operations that it (adaptively) makes on any $v\in \mathcal V$. The primitive complexity of a menu is the minimal primitive complexity of any algorithm that for any valuation $v\in \mathcal V$ computes a most profitable bundle in this menu.\footnote{As stated, this definition considers deterministic menus (as the menu entries are bundles). This definition naturally extends to randomized menus in which an entry might be a lottery over bundles.}
If the algorithm is randomized, the randomized primitive complexity is the \emph{expected} number of queries that the algorithm makes, where expectation is over the internal random coins of the algorithm. 
Similarly, if the valuations are drawn from some distribution, the distributional primitive complexity is the expected number of queries that the algorithm makes, where expectation is over valuations sampled from the prior distribution. 
We note that our randomized algorithms will always return a correct solution, not just with high probability. Randomization is only used to reduce the (expected) number of queries. In contrast, all of our lower bounds hold even
for algorithms that only succeed with constant probability.
	\section{Algorithms for Maximization Subject to a Cardinality Constraint}\label{sec:algorithms}

In this section we first present an algorithm that finds a $k$-optimal set for any additive valuation, and then present an algorithm that does the same for any weighted matroid rank valuation.
Our algorithms are randomized and make only $poly\log(m)$ queries in expectation. Note that this is an exponential improvement in the number of queries over the ``obvious'' algorithms: the trivial algorithm for additive valuations finds a $k$-optimal set with $m$ queries (one for each item). For weighted matroid rank valuations, a $k$-optimal set can be found by running the greedy algorithm ($poly(m)$ value queries).

We start with the case of additive valuations. 
We then solve the more general case of weighted matroid-rank valuations by 
first finding a maximal weight independent set $R$ using $poly\log(m)$ queries, and then applying the algorithm for additive valuations as a sub-procedure.

\begin{theorem}\label{thm:findmaxk}
There exists a randomized algorithm that for any additive valuation $v$ over $m$ items, finds a \opt{k} set using value and demand queries, and in expectation makes $\O{\log^3 m}$ queries.
\end{theorem}

We further extend the result to weighted matroid-rank valuations.

\begin{theorem}\label{thm:findmaxkMR}
There exists a randomized algorithm that for any weighted matroid-rank valuation $v$ over $m$ items, finds a \opt{k} set using value and demand queries, and in expectation makes $\O{\log^3 m}$ queries.
\end{theorem}
Proofs of these theorems can be found in Appendix \ref{sec-additive} and Appendix \ref{sec-mr}. We now provide some intuition for the case of additive valuations. We start with presenting an algorithm for the special case in which all items have distinct values. That is, for each two items $i\neq j$ we have that $v_i \neq v_j$. 
In this case, we can easily find the set of $k$ highest value items with $\O{\log m}$ queries: select an item uniformly at random and denote the value of this item by $v$. Make a uniform demand query with a price of $v$ per item. Let $D$ be the returned demanded set and denote $m'=|D|$.
 If $k\geq m'$ we know that all the $m'$ items of $D$ are among the $k$ highest values, so we pick them and remove them from the set, update the number of items we still need to pick to $k'=k-m'$, and continue recursively on the remaining items that are not in $D$ to select an additional $k=k'$ items. 
 If $k<m'$ we similarly remove the items not in $D$, and continue recursively, aiming to pick $k$ out of the items in $D$. 
 The expected number of iterations is $\O{\log m}$  since in each iteration, in expectation, half of the items are removed from consideration (either picked or discarded).

The problem is more challenging if the values are not distinct. That is, if there are items of equal value ($v_i=v_j$ for items $i, j\neq i$) and we make a demand query with price $p = v_i=v_j$ per item. 
The challenge is that, unless assuming a specific tie-breaking rule (which we do not), the demanded set might contain all items of value $p$, none of them, or some arbitrary subset of them.
We present an algorithm that works for any implementation of the demand query, and even if the tie breaking is adversarial.

\section{The Primitive Complexity of Bundle-Size Pricing} \label{sec:symmenu}

Recall that a {\em bundle-size pricing menu} is a set of offers $\{(q_i,p_i)\}_i$, each of the form ``pay $p_i$ and choose any set of size $q_i$ of items to receive". 
We call the number of different quantities that are offered the \emph{size} of the bundle-size pricing menu. 
Given a bundle-size pricing menu, a buyer that wants to find a profit maximizing set can do so by finding a \opt{q_i} set for each of quantity $q_i$ specified by the menu, and picking the one of highest profit among the candidates sets. 
As with $m$ items there can be $m$ different quantities specified, this approach will result in polynomial (in $m$) number of \opt{k} set problems that need to be solved, and thus  
require at least polynomial number of demand queries. We aim for sub-linear number of queries.

We first observe, using a variant of a result of \citet{hart2019selling}, that any bundle-size pricing menu can be transformed to another bundle-size pricing, losing only $\epsilon$-fraction of the revenue, but with the size of the new menu only depending on the revenue loss parameter $\epsilon$ and on the largest ratio of prices in the original bundle-size pricing menu, but not on $m$.
Specifically, the new bundle-size pricing menu size will only be polynomial in $1/\epsilon$ and in the logarithm of the maximal ratio of prices in the menu, but independent of the number of items $m$. 
For valuation classes for which the \opt{k} set problem is solvable in poly-logarithmic number of demand queries (as additive and weighted matroid-rank valuations), we can thus derive a bound on the number of demand queries needed to pick a profit-maximizing bundle from this smaller bundle-size pricing menu. 
Yet, the number of demand queries in above result depends on the maximum price-ratio not being too large, and will be polynomial in $m$ when this ratio is exponential. 

Our main result in this section is that we can get rid of this limitation when optimizing the \emph{expected} revenue for a given distribution $\D$ over valuations (rather than ex-post, for any given valuation).  
{We consider the expected revenue of the menu for the given distribution $\D$ and  further shrink the menu size by removing some of the priced bundles sizes. Specifically, we remove each bundle size that contributes at most $\frac{\epsilon}{m}$ fraction of the expected revenue (over $\D$), losing another $\epsilon$ fraction of the total revenue. 
Yet, even after this additional processing step that reduces the size of the menu, the menu size might still be large (not poly-logarithmic). Thus, it might well still be the case that the size of this bundle-size menu is  not small enough to 
get a poly-logarithmic number of queries by always finding a profit-maximizing bundle for each size and taking the best one. Nevertheless, we show that either the same revenue can be obtained by replacing the menu with a menu that only sells the grand bundle, or that for the same menu a
poly-logarithmic number of queries is sufficient \emph{in expectation}.
That is, although the menu might contain many distinct prices and bundles sizes, only $poly(\log m,\frac{1}{\epsilon})$ bundle sizes need to be considered (in expectation over the distribution $\D$) to find a profit-maximizing bundle for a valuation $v$ sampled from $\D$.
} 
The proof is in Appendix \ref{app:imp}.  


\begin{theorem}\label{thm:menu_expected}
	Given a distribution $\D$ over weighted matroid-rank valuations over a set $\items$ of $m$ items, a bundle-size pricing menu $\menu_1$, and $\epsilon>0$, there exists a bundle-size pricing menu $\menu_2$ such that $\REV(\menu_2,\D)\geq(1-\epsilon)\REV(\menu_1,\D)$ and such that a profit-maximizing bundle for $\menu_2$ can be found in $\O{\epsilon^{-3}\log^4 m}$ value and demand queries in expectation,  where the expectation is over the internal random coins of the algorithm and the distribution $\D$. 
\end{theorem}

	\section{Hardness of Maximization Subject to a Cardinality Constraint}\label{sec:lb}
We have presented a randomized algorithm  that for {weighted matroid-rank} 
valuations  finds a \opt{k} set using poly-logarithmic number of demand and value queries
(Theorem \ref{thm:findmaxkMR}).
In this section we present lower bounds for several related problems, showing that our results can not be strengthen in multiple ways. 
Namely, we show that:
\begin{itemize}
	\item For \textbf{submodular} valuations, any randomized algorithm that for every valuation succeeds with constant probability to find a \opt{k} set, must use, in expectation, an exponential number of value and demand queries (Section \ref{subsec-exponential-sm}). 
	\item Any \textbf{deterministic} algorithm that finds a \opt{1} set (an item with the highest value) requires $\Omega(\sqrt{m})$ value and demand queries, even when the valuation is additive (Section \ref{sec:det_lb}). 
    \item For additive valuations, we show, using two proof techniques, that any algorithm that given $k$ finds a \opt{k} set and succeeds with constant probability requires $\Om{m/\log m}$ \textbf{value} queries, even if randomization is allowed. 
    Furthermore, for deterministic algorithms we show that $m-1$ {value} queries are needed. These results are presented in Section~\ref{sec:value_lb}.
\end{itemize}

\subsection{An Exponential Lower Bound for Submodular Valuations}\label{subsec-exponential-sm}


We next consider submodular valuations and show 
that any randomized algorithm that with value and demand queries finds (with a constant probability) a \opt{k} set makes in expectation exponentially many queries:
\begin{theorem}\label{thm:submodular}
    {Let $A$ be a randomized algorithm that given a submodular valuation and $k$ finds a $k$-optimal set  by making value and demand queries.} For $m$ that is large enough, if $A$ succeeds with probability at least $\frac{1}{2}$ then $A$ makes at least $1.3^m$ queries. 
\end{theorem}
This solves an open question of \cite{badanidiyuru2012optimization} that provided a $\frac 9 8$-approximation for this problem, but did not even rule out the possibility that a $k$-optimal set can be found with a polynomial number of value and demand queries.

The rest of this subsection is devoted to outlining the proof of the theorem. To prove a bound for randomized algorithms, it is enough to provide a distribution $\D$ over valuations such that the probability that a deterministic algorithm that makes subexponentially many value and demand queries finds a $k$-optimal set in a valuation that is sampled from $\D$ is small, by Yao's principle. 

We will prove the theorem for an even $m$ and $k=\frac m 2$. Let $\D$ be the following distribution over submodular valuations over a set of items $\items$: 
each valuation $v$ is defined by a family $\Bv$  of sets of size $k+1$ and a set $\Gv$ of size $k$ ($\Gv$ will be the $k$-optimal bundle). Each set of size $k+1$ is included in $\Bv$ with probability $\frac 1 {m^{2}}$, independently at random. Out of the sets of size $k$ that are not contained in any of the sets in $\Bv$, we choose one random set and denote it by $\Gv$.\footnote{There is an exponentially small probability that  every set of size $k$ is contained in some set of $\Bv$. In this case $\Gv$ is not defined and all sets of size $k$ have the same value. We thus condition our analysis on having that this event does not happen and that $\Gv$ is defined.} It will also be convenient to define $\Rv$ to be the family of all sets of size $k-1$ that are not contained in any of the sets in $\Bv$. The valuation $v$ is then defined as follows:

\begin{equation*}
\label{eq:sm_definition}
v(S) = 
\begin{cases}
k	 	& |S|>k+1\\
k	 	& |S|=k+1,~S\in\Bv\\
k-3/11 	& |S|=k+1,~S\notin\Bv\\
k-6/11 	& |S|=k,~S=\Gv\\
k-7/11 	& |S|=k,~S\neq \Gv\\
k-1 	& |S|=k-1,~S\in\Rv\\
k-14/11 & |S|=k-1,~S\notin\Rv\\
|S|	 & |S|<k-1
\end{cases}
\end{equation*}
The set $\Gv$ is the \opt{k} set. Roughly speaking, sets from the families $\Bv$ and $\Rv$ guarantee that it is very unlikely that any information about the identity of $\Gv$ will be provided by any demand or value query.
See Figure~\ref{fig:sm} for an illustration of the relations between the sets. 
\begin{figure}
	\centering
	\includegraphics[width=0.7\linewidth]{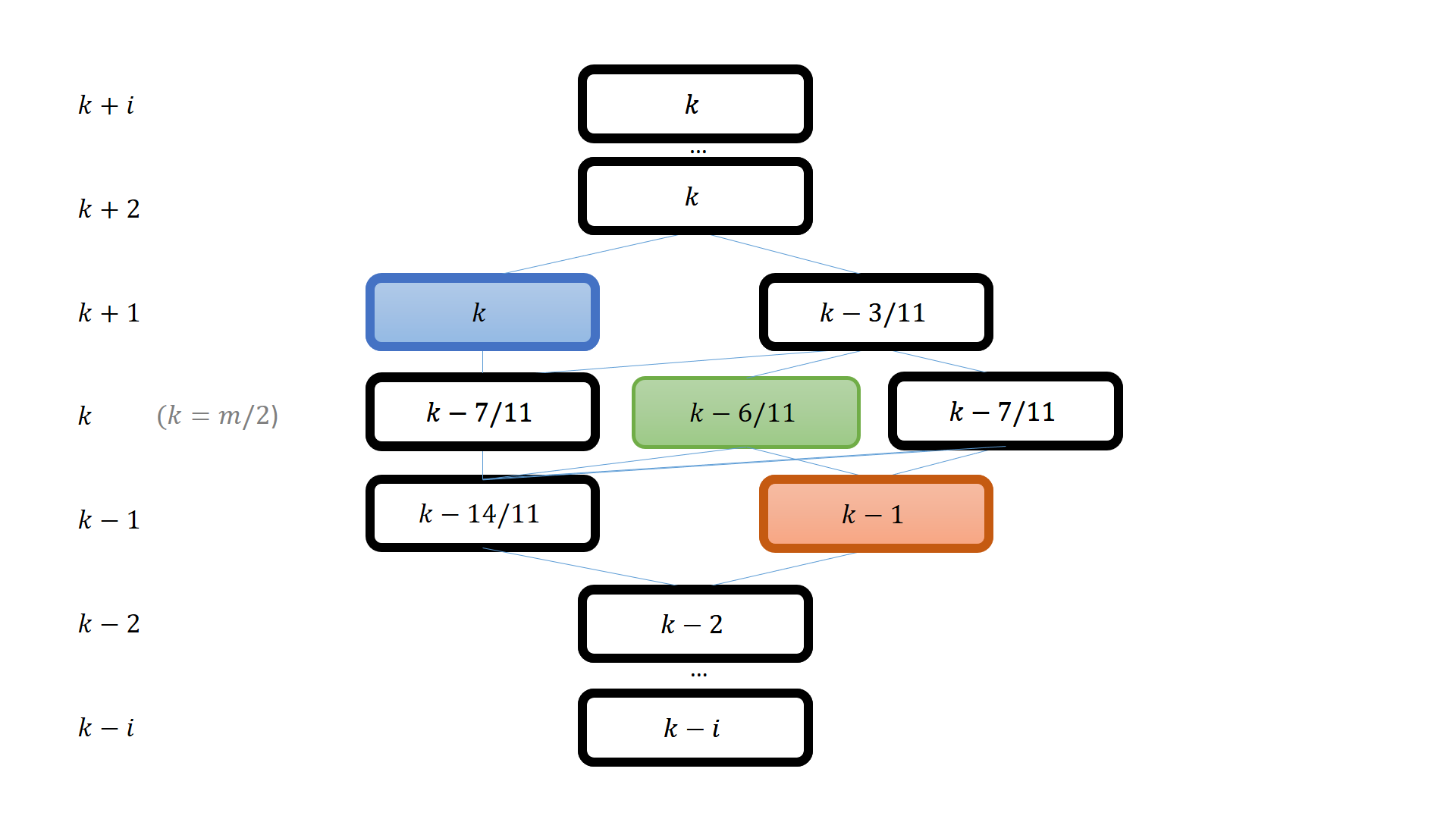}
	\caption{The relation between bundles in valuations in the support of $\D$. Blue squares denote bundles that are in $\Bv$, the green square denote the bundle $\Gv$, and red square denote the bundles in $\Rv$. A line between squares denotes a possible containment relationship between the bundles. Note that all subsets of size $k-1$ of a set in $\Bv$ have value $k-14/11$.}.
	\label{fig:sm}
\end{figure}


We prove our exponential lower bound in two steps.
We  onsider any algorithm that uses only value queries and on valuation that is samples from $\D$ finds a \opt{k} set with non-negligible probability. 
We show that any such algorithm makes, in expectation, an exponential number of value queries (Lemma \ref{lem:SM_value_queries}). 
We complete the proof by showing that, with high probability over $\D$, all demand queries on a valuation $v$ sampled from $\D$ can be simulated by value queries with only a polynomial blowup in the number of queries
(Lemma \ref{lem:SM_demand2value}). 
\begin{lemma}\label{lem:SM_value_queries}
    Fix some deterministic algorithm $A$ that makes only value queries and the set of those queries is in a canonical form\footnote{Later (Appendix \ref{app:submodular}), we formally define what it means for a set of value queries to be in a canonical form. We comment now that any set value queries of size $t$ can be converted to a canonical form by making $poly(t,m)$ additional value queries.}. Suppose that $A$ makes $t<1.9^m$ value queries on valuations that are sampled from $\D$. Then, for a large enough $m$, the probability (over $\D$) that $A$ finds a $k$-optimal set is at most $\frac {t} {{1.9}^{m}}$.
\end{lemma}
\begin{lemma}\label{lem:SM_demand2value}
    Fix a deterministic algorithm $A$ that uses $t$ demand and value queries for valuations sampled from $\D$. For any $\alpha>1$, with probability $1-\frac 1 \alpha$, $A$ can be implemented using at most $2m^5\cdot t^2\cdot \alpha$ value queries.
\end{lemma} 
The proof of Lemma~\ref{lem:SM_value_queries} is presented in Appendix \ref{subsec-submodular-value}. 
Before proving Lemma~\ref{lem:SM_demand2value} in Section \ref{subsec-submodular-demand}, we present some definitions and auxiliary claims. 

\subsubsection{Definitions and Auxiliary Claims}\label{subsec-submodular-aux}
We first show that every valuation $v$ in the support of $\D$ is indeed submodular. We then present some definitions and prove several claims that will be helpful in the proof of the theorem.
All proofs in this section are deferred to Appendix~\ref{subsec-submodular-aux-app}.

\begin{lemma}
	Every valuation $v$ in the support of $\D$ is submodular.
\end{lemma}
A valuation in the support of $\D$ is completely defined by the values of all sets of size $k$ and $k+1$.
We say that a set of value queries $\Q$ is in a {\em canonical form} if all queries in $\Q$ are for sets of size $k$ or $k+1$, and for every query of size $k$ all of its supersets of size $k+1$ are also in $\Q$. 
{Essentially, 
all information that a set of value queries conveys  about a valuation can also be conveyed by some set of queries that is in a canonical form and is not much larger.}
The next proposition shows that we can assume that the query set is in a canonical form at a cost of a polynomial blow-up in the number of queries:

\begin{proposition}
Let $A'$ be an algorithm that makes $t$ value queries on a valuation  in the support of $\D$. Then, there is an algorithm $A$ that simulates $A'$ while making $m^2\cdot t$ value queries on a valuation in the support of $\D$. Moreover, the set of queries that $A$ makes has a canonical form.
\end{proposition}

We next present several useful definitions and notations.  
Fix some deterministic algorithm $A$ that makes only value queries and runs on valuations from $\D$. 
Fix any valuation $v$ from the support of $\D$, and let $\Qv$ denote the list of $t$ bundles that $A$ queried together with their values.
Let $\DQv$ denote the distribution over valuations that is obtained by sampling according to $\D$ a valuation that is consistent with the queries in $\Qv$. 
    Let $\Byes$ be the family of sets that includes every set $S$ such that $\Pr_{v'\sim\DQv}[S\in\Bvp]=1$.
    Similarly, let $\Bno$ be the family of sets that includes every set $S$ such that $\Pr_{v'\sim\DQv}[S\in\Bvp]=0$.
    Let $\Qvk$ be the family of sets of size $k$ that were queried in $\Qv$.
 
We now claim that assuming queries are in a canonical form, the conditional distribution for sets not queried is essentially identical to the prior.
\begin{lemma}\label{lem:B_is_independent}
    Fix any valuation $v$ sampled from $\D$ and assume $\Qv$ is in a canonical form. 
    It holds that:
    \begin{itemize}
        \item For any set $S$ of size $k+1$ it holds that $\Pr_{v'\sim\DQv}[S\in\Bvp]\in \{0,1,\frac 1 {m^{2}}\}$.
        \item The conditional probabilities are independent:  for any family $\mathcal{F}$ of sets of size $k+1$ it holds that
    $\Pr_{v'\sim\DQv}[\forall S\in\mathcal{F}\ ,S\notin\Bvp]= \prod_{S\in\mathcal{F}}\Pr_{v'\sim\DQv}[S\notin\Bvp]$.
    \end{itemize}
\end{lemma}

\subsubsection{Proof of Lemma~\ref{lem:SM_demand2value}: Simulating Demand Queries by Value Queries}\label{subsec-submodular-demand}
In this section we prove Lemma~\ref{lem:SM_demand2value}, showing that for valuations drawn from $\D$, demand queries can be simulated by value queries, {and with high probability polynomial number of queries is sufficient  for the simulation}.

\cout{

\begin{lemma}\label{lem:SM_demand2value}
    
    Fix a deterministic algorithm $A$ that uses $t$ demand and value queries for valuations sampled from $\D$. For any $\alpha>1$, with probability $1-\frac 1 \alpha$, $A$ can be implemented using at most $2m^5\cdot t^2\cdot \alpha$ value queries.
\end{lemma}
}

\begin{proof} [Proof of Lemma~\ref{lem:SM_demand2value}]
   First, we prove by induction on $t$ that a set of value queries in a canonical form that is followed by a demand query, can be simulated by a set of value queries that has a canonical form, and the expected size of that set is at most $t+2\cdot m^5$. 
   The claim trivially holds for $t=0$ and Claim \ref{claim:induction-step} proves the induction step.
   Second, given the claim, we use Markov's inequality, to argue that 
   for any $\alpha>0$, with probability at most $\frac{1}{t\cdot \alpha}$, more than $2m^5\cdot t\cdot \alpha$ value queries are needed for the implementation of a query.  Hence, using the union bound, with probability $1-\frac{t}{t\cdot \alpha}=1-\frac 1 \alpha$ all $t$ demand queries can be implemented using at most $2m^5\cdot t^2\cdot \alpha$ value queries.



    \begin{claim}\label{claim:induction-step}
    Fix a deterministic algorithm $A$ that runs on valuations sampled from $\D$ and up to some point has used a set of $t$ value queries, and the set has a canonical form. 
    Fix any demand query for price vector $p$. 
    Then, it is possible to simulate all these queries (including the demand query) by a set of value queries that has a canonical form, and the expected size of that set is at most $t+2\cdot m^5$. 
    \end{claim}
    \begin{proof}
    For $m\leq 3$, the total number of subsets of $2^m$ is smaller than $2\cdot m^5$ and the claim trivially holds. We next assume that $m$ is even and $m\geq 4$. 

    We first present some intuition for the proof. 
    Consider a demand query with price vector $p$ for valuation $v$.
	For a set $S$, denote by $U(S,p)=v(S)-p(S)$ the profit from buying set $S$ at price $p(S)=\sum_{i\in S} p_i$. A demand query returns  a set $S$ that has maximal profit under price vector $p$.  
	Clearly, if we can find a most profitable bundle of size $d$ for every $d\in[m]$ then we can return a most profitable set (with a set from these $m$ bundles that is most profitable). 
	While finding a most profitable bundle of every size is clearly sufficient, it turns out it is not necessary, and we show that a most profitable set can be found with polynomially many value queries, without always knowing a most profitable bundle of size $k$. 
	We first show that for each $d\neq k$, a most profitable bundle of size $d$ can indeed be found by value queries to a family $\Fd$ of sets of size $d$ that we can specify.
	Second, we show that if every set that is a most profitable set overall is of size $k$, then such a set can also be found by value queries to a family $\Fk$ of sets of size $k$ that we can specify.
	Finally, we show that the set of all queries (the $t$ value queries as well as value queries to new sets that are in these families of sets) is only polynomially larger than $t$. We next present the formal claim and its proof.
        
    For each $d\in [m]$ we define a family $\Fd$ of sets of size $d$, such that: 
    \begin{itemize}
        \item If $d\neq k$ then some set of size $d$ that has the highest profit among all sets of size $d$ in the family $\Fd$.
        \item If every demanded set is of size $d=k$, then 
        a demanded set of size $k$ belongs to the family $\Fk$.
        \item For even $m\geq 4$ it holds that $\E[v\sim D]{\sum_{d\in[m]}|\Fd\setminus\Qv|} \leq 2m^3$.
    \end{itemize}

Assume algorithm $A$ is running on valuation $v$ sampled from $\D$, and the algorithm was using the set $\Qv$ of value queries that is a canonical form.  For each size $d$ we consider the list of size $d$ from cheapest to most expensive (breaking ties arbitrarily). Let $\Sdp$ denote a cheapest set of size $d$. 
    \begin{itemize}
        \item For $d<k-1$ or $d>k+1$, all bundles of size $d$ have the same value, thus a profit maximizing set of size $d$ is simply some cheapest set of size $d$, so we define  $\Fd=\{\Sdp\}$. 
        \item For $d=k+1$, there are two possible values for a bundle of size $k+1$, 
        depending on whether the bundle is in $\Bv$ or not. 
        Let $S_{\Bv}$ be a cheapest bundle in $\Bv$ that is the first in order of set prices.
        A most profitable set of size $d=k+1$ is then either $S_{\Bv}$ or $\Spkp$. 
        We add to $\Fkp$ the cheapest sets of size $d=k+1$ in increasing order of price, till we find the first set that belongs to ${\Bv}$. Note that $\Spkp$ is the first added set and is always in  $\Fkp$. 
        Since $Q_v$ is in a canonical form, by Lemma \ref{lem:B_is_independent},
        it holds that either $S\in \Q_v$ or that the probability that $S\in {\Bvp}$ conditional on $v'$ being sampled according to $\DQv$ is $\frac {1} {m^{2}}$. 
        Hence the expected number of cheapest bundles of size $k+1$ that are not in $\Qv$ till a set in $\Bv$ is found is at most $m^2$, that is, $\E{|\Fkp\setminus\Qv|}\leq m^2$.
        
        \item For $d=k-1$, there are two possible values for a bundle of size $k-1$, 
        depending on whether the bundle is in ${\Rv}$ or not. 
        Let $S_{\Rv}$ be a cheapest bundle in ${\Rv}$ that is the first in order of set prices.
        A most profitable set of size $d=k-1$ is then either $S_{\Rv}$ or $\Spkm$. 
        We add to $\Fkm$ the cheapest sets of size $d=k-1$ in increasing order of price, till we find the first set that belongs to ${\Rv}$. Note that $\Spkm$ is the first added set and is always in  $\Fkm$. 
        
        Consider some set $S$ of size $k-1$, we need bound the probability that $S$ is in ${\Rvp}$ given that $v'$ is sampled according to $\DQv$. 
        For given $v'$, $S\in \Rvp$ if none of its supersets are in $\Bvp$. 
        Hence, by Lemma \ref{lem:B_is_independent}, since $\Qv$ is in a canonical form, 
        for any set $S$ of size $k-1$ we have that either $\Pr_{v'\sim\DQv}[S\in{\Rvp}]=0$ or $\Pr_{v'\sim\DQv}]S\in\Rvp]\geq (1-\frac 1 {m^{2}})^{(m-k)(m-k-1)/2}>(1-\frac 1 {m^{2}})^{\frac{m^2-1}{2}}>e^{-0.5}>0.5$.
        Hence the expected number of cheapest bundles of size $k-1$ that are not in $\Qv$ till a set in $\Rvp$ is found (for $v'\sim\DQv$) is at most $2$, that is, $\E{|\Fkm\setminus\Qv|}\leq 2$.
        
        \item For $d=k$, there are two possible values for a bundle of size $k$, 
        depending on whether the bundle is $\Gv$ or not. 
        A most profitable set of size $d=k$ is either $\Gv$ or $\Spk$. 
        {In Claim \ref{claim:sizek}} we show that when it is not $\Spk$ then it must be either a set of size $k$ that is a subset of a set in $\Fkp$, or a set of size $k$ that is a superset of a set in $\Fkm$.     
        Thus we define $\Fk$ to include $\Spk$, all sets of size $k$ that are a subset of a set in $\Fkp$, and all  sets of size $k$ that are a superset of a set in $\Fkm$.
        Note that the expected size of $\Fk$ satisfies 
        $
        \E[v\sim D]{|\Fk\setminus\Qv|}\leq \E[v\sim D]{|\Fkm\setminus\Qv|\cdot (m-k)} +
        \E[v\sim D]{|\Fkp\setminus\Qv|\cdot (k+1)} +1  $.
        
    \end{itemize}

    In total, 
    \begin{equation*}
    \begin{split}
    \E[v\sim D]{\sum_{d\in[m]}|\Fd\setminus\Qv|} \leq & (m-3) + \E[v\sim D]{|\Fkm\setminus\Qv|\cdot (m-k+1)} +\E[v\sim D]{|\Fkp\setminus\Qv|\cdot (k+2)} +1\\
    \leq & m + 2(m-k+1) + m^2(k+2) = 
    m^3/2 + 2m^2 +2m +  2 \leq 2m^3
    \end{split}
    \end{equation*}
    as $k=\frac{m}{2}$ and $m\geq 4$.
    
    Thus, each demand query can implemented by $2m^3$ value queries in expectation.
    As we mentioned earlier, moving to a canonical form requires replacing each value query with at most $m^2$ value queries. Therefore, the demand query can be implemented while remaining in a canonical form using $2m^5$ value queries in expectation.
    
    We now complete the proof by showing that for $d=k$, the family $\Fd$ contains a most profitable bundle whenever the most profitable bundle is of size $k$.
    \begin{claim}\label{claim:sizek}
        For any price vector $p$, if for valuation $v$ every demanded set is of size $k$, then
        any most profitable set (of size $k$) is either $\Spk$ (cheapest set of size $k$), a set of size $k$ that is a subset of a set in $\Fkp$, or a set of size $k$ that is a superset of a set in $\Fkm$.
    \end{claim}
    \begin{proof} 
    Assume that every demanded set is of size $k$.  Fix any set that most profitable set of size $k$ and denote it by $\Upk$.
        If $\Spk= \Gv$ then it must be that $\Upk=\Gv$, that is, $\Gv$ must be the unique most profitable set of size $k$ (as $\Gv$ has higher value than any other set of size $k$). 
        So we can assume that $\Spk\neq \Gv$. If $\Upk\neq \Gv$ then $\Upk$ must be $\Spk$. We thus assume that $\Upk=\Gv$ (and $\Spk\neq \Gv$), and the value of $\Upk$ is thus $k-6/11$. 
    
        The family $\Fkm$ includes a set $R$ from $\Rv$ with value $k-1$. 
        As $\Upk=\Gv$ is more profitable than $R$, it holds that $v(\Gv)-p(\Gv) = k-6/11 - p(\Gv) > k-1-p(R)$ and thus $p(\Gv)-5/11<p(R)$. 
        If there is an item $i\in \Gv$ of price at least $5/11$ then the set $\Gv\setminus \{i\}$ has price smaller than $p(\Gv\setminus\set{i})\leq p(\Gv)-5/11<p(R)$, and thus the set $\Gv\setminus \{i\}$ is in $\Fkm$ which implies that $\Gv\in \Fk$ as needed.
        
    
        The family $\Fkp$ includes a set $B$ from $\Bv$ with value $k$. 
        As $\Upk=\Gv$ is more profitable than $B$, it holds that $v(\Gv)-p(\Gv) = k-6/11 - p(\Gv) > k-p(B)$ and thus $p(B)-p(\Gv)>6/11$. 
        If there is an item $i\notin \Gv$ of price lower than $6/11$ then the set $\Gv\cup \{i\}$ has price smaller than $p(\Gv)+6/11<p(B)$, and thus the set $\Gv\cup \{i\}$ is in $\Fkp$ which implies that $\Gv\in \Fk$ as needed. 
        
        Otherwise, the price of every item in $\Upk$ is less than $\frac{5}{11}$, and the price of every item not in $\Upk$ is more than $\frac{6}{11}$,  and thus $\Upk=\Spk$ is the unique cheapest bundle of size $k$, a contradiction to $\Gv=\Upk\neq \Spk$.
       \end{proof}

    This completes the proof of Claim  \ref{claim:induction-step}.   
    \end{proof}
    
    
This completes the proof of Lemma \ref{lem:SM_demand2value}.
\end{proof}


We now conclude the proof of Theorem \ref{thm:submodular}. By Lemma~\ref{lem:SM_demand2value}, we have that with probability $1-\frac 1 \alpha$ over $\D$, a deterministic algorithm that makes $t$ demand and value queries can be implemanted using $2m^5\cdot t^{2}\cdot\alpha$ values queries in a canonical form.
Let $t'=2m^5\cdot t^{2}\cdot\alpha$.
By lemma~\ref{lem:SM_value_queries}, implementation that uses a set of $t'$ value queries that is in a canonical form, for $t'<1.9^m$ and large enough $m$, has a probability of at most $\frac{t'}{1.9^m}$ for finding $\Gv$. 
Hence, the original algorithm fails with probability at least $1-\alpha^{-1}-\frac{t'}{1.9^m}$.
Taking $t=1.3^m$, $\alpha=3$, and $m$ large enough, we have that a deterministic algorithm that makes at most $1.3^m$ queries, fails with probability at least $1-\frac{6\cdot  m^5\cdot 1.3^{2m}}{1.9^m}-\frac{1}{3}>\frac{1}{2}$ over $\D$.
\subsection{An $\Omega(\sqrt{m})$ Deterministic Lower Bound for Additive Valuations}\label{sec:det_lb}
In this section we show that randomization is inherently required for maximizing an additive valuation subject to a cardinality constraint. We show that every deterministic algorithm that always finds an item with the smallest value (equivalently, finds an \opt{(m-1)} set) must make $\sqrt{m}-1$  value and demand queries. We prove the following theorem:

\begin{theorem}\label{thm-deterministic}
Let $A$ be a deterministic algorithm that for any additive valuation finds an $(m-1)$-optimal set using value and demand queries.
Then, $A$ makes at least $\Omega(\sqrt m)$ 
queries.
\end{theorem}


Since we discuss only additive valuations in this section, we abuse notation and sometimes refer to additive valuations as vectors in $\R^m$, with each element representing the value of the corresponding item. We start with several lemmas regarding linear constraints that will be useful in the proof. First, recall that {\em basic feasible solutions (BFS)} are non-negative solutions to a linear system with minimal support (see, e.g., \cite{eisenbrand2006caratheodory,matousek2007understanding}).

\begin{lemma}[\cite{matousek2007understanding}]\label{lem:sparse}
	For $A\in \R^{\ell\times m}$ and $b\in \R^{\ell}$, if the system ($A\cdot v=b$, $v\geq 0$) has a solution, it has a solution with support of size at most $\ell$.
\end{lemma}


The following simple observation will be useful later in the proof. 
\begin{claim}\label{pro:continuance}
    Let $\alpha,\beta\in \R$ such that $\alpha+\beta=1$.
	If $\vec x,\vec y$ are solutions of the linear system $A\cdot v=b$, then so is $\vec z = \alpha\cdot \vec x + \beta \cdot \vec y$.
\end{claim}
\begin{proof}
Since both $\vec x$ and $\vec y$ are solutions of the linear system $A\cdot v=b$, we have that 
$$A\cdot \vec z=A\cdot (\alpha\cdot\vec  x + \beta\cdot \vec{y})=\alpha \cdot A\cdot\vec  x +\beta \cdot A \cdot \vec{y} =(\alpha+\beta)\cdot b=b$$
\end{proof}

\begin{lemma}\label{lem:n_value_queries}
    Let $A\in \R^{\ell\times m}$ be a matrix of rank at most $m-1$.
	The linear system $A\cdot v=A\cdot\vec{\textbf{1}}$ with the constraint $v\geq 0$ has two different solutions which disagree on the identity of the item with the smallest value.
\end{lemma}
\begin{proof}
	The vector $\vec{\textbf{1}}$ is a solution. In addition, since $A$'s rank is smaller than $m$, there exists some solution $\vec w\in\R^m$ such that $\vec w$ and $\vec{\textbf{1}}$ are independent. Let $\vec x = (1-\epsilon)\vec{\textbf{1}} + \epsilon\cdot \vec w$ and $\vec y = (1+\epsilon)\vec{\textbf{1}} - \epsilon\cdot \vec w$ for a small enough $\epsilon>0$ to ensure that all entries in $\vec x$ and $\vec y$ are positive. 
	Note that by Claim \ref{pro:continuance} both $\vec x$ and $\vec y$ are solutions of the linear system.
	The values in $\vec{x}$ and $\vec{y}$ are ordered according to $\vec{w}$ where a minimal item in $\vec{x}$ is maximal and $\vec{y}$ and vice versa. 
	Since $\vec{w}$ is not the all zero vector, these two solutions disagree on the identity of the items with the smallest value.
\end{proof}


We are now ready to prove Theorem \ref{thm-deterministic}.
Throughout the proof, we assume that in all demand queries that $A$ makes the prices are strictly positive.
This assumption only doubles the number of queries:
if there is a demand query that gives $0$ prices for some set of items $S$, we can replace all $0$ prices with $\infty$, run the new demand query and return the union of $S$ and the answer $T$ of the new demand query. Since the demand query should also return the value of $S\cup T$ we can make one value query $v(S)$ and return $v(S)+v(T)$ (recall that a demand query also returns the value of the most demanded set). 
{In addition, we assume that the first query is a value query for the entire set $\items$, increasing the total queries made by the algorithm by at most one.}
As the algorithm must work with any implementation of the demand query, it must work with the one we specify here.

\begin{proof}[of Theorem \ref{thm-deterministic}]
We show that for every deterministic algorithm $A$ there is an adversary that can answer all queries in a way such that as long as the algorithm did not make many queries, there are two different valuations $v,v'$, both consistent with the queries asked, each has a unique item with minimal value, but the items with the minimal value in $v$ and in $v'$ are different.
Thus, the algorithm does not distinguish between $v$ and $v'$ and does not find an $(m-1)$-optimal bundle.

We will describe a set  of valuations and inductively show that every valuation in this set is consistent with the queries asked so far. Thus, every item that is a minimal item of a valuation in the set is a possible solution. 

The next claim is the heart of the proof. To give some intuition, let us examine some of the possible answers of the adversary. For every bundle $S$ that its value is queried, the algorithm will return the value $|S|$. Suppose that a demand query, all with positive prices, is made. The adversary now ``defines'' a set of items $L$ that we ``set'' their value to be very big. The items in $L$ are the items that the demand query returns. We consider the values of items that are not in $L$ to be very low, although the adversary does not commit on their specific values (so any of them might be the minimal item). In the next value queries we will treat every $S\subseteq M-L$ as having value $v(S)=\eps'\cdot |S|$, for $\eps \gg \eps'$ ($\eps'$ will be smaller than the smallest price in the demand query). Thus, the demand query will only return the items in $L$. 

There are several challenges in achieving this. The first is making sure that after the first demand query we are consistent with the value queries that were done so far. The second is to be able to answer not just the first demand query but also the following ones. The third is to make sure that the set of items $L$ is small, otherwise the adversary has to commit on the values of all items too quickly. The next claim handles all these challenges.

\begin{claim}\label{claim-invariant}
Consider an execution of the algorithm after a set of $\mathcal Q$ queries. 
Let $t_{r}$ denote the number of value queries made before the $r$'th demand query in $\mathcal Q$. Suppose that $\mathcal Q$ contains $i$ demand queries and $x$ value queries.
Let $x'$ denote the number of value queries made after the last ($i$'th) demand query. There exists an adversary, a non-empty set of valuations $\mathcal V^{\mathcal Q}$, a set of items $\LQ\subseteq \items$ where $|\LQ|=\sum_{r\leq i}t_{r}$, a matrix $\AQ\in \set{0,1}^{(t_i+x')\times (m-|\LQ|)}$ and $\eQ>0$ such that:
\begin{itemize}
    \item All valuations in $v\in \mathcal V^{\mathcal Q}$ are consistent with all the queries made so far.
    \item For every $v,v'\in \mathcal V^{\mathcal Q}$ and $j\in \LQ$, $v(\{j\})=v'(\{j\})$.
    \item By renaming, assume without loss of generality that the items that are not in $\LQ$ are indexed $1,\ldots, m-|\LQ|$. For each $v\in \mathcal V^{\mathcal Q}$ it holds that:
        \begin{equation*}
            \AQ\cdot v = \eQ\cdot \AQ\cdot \vec{\textbf{1}}
        \end{equation*}
\end{itemize}
\end{claim}
\begin{proof}
{We prove the claim by induction over the number of queries. The base case is when $\mathcal Q$ includes a single query, a value query for $\items$ that returns a value of $m$. Let $\LQ=\emptyset$, $\eQ=1$, and $\AQ$ which is a single all-one line. The claim trivially holds for the case the valuation $\vec{\textbf{1}}$ satisfies all three requirements.
}

We now assume the induction hypothesis for $\mathcal Q'$, and prove for $\mathcal Q$, where $\mathcal Q$ is $\mathcal Q'$ with an additional query. We describe how the adversary answers the additional query and how to obtain $\AQ$, $\LQ$, $\eQ$, for which the claim holds after each query. For each item $j\in \LQp$, let $\LQp(j)$ be the value of item $j$ for all valuations in $\mathcal V^{\mathcal Q'}$.

We first consider the case where the additional query is a value query, and then the case where the additional query is a demand query. If the algorithm makes a value query for some set $S$, the adversary answers that the value of $S$ is $\eQp\cdot |S\setminus \LQp|+\sum_{j\in S\cap \LQp} \LQp(j)$. This answer is consistent with the previous queries since there exists a valuation $v\in \mathcal V^{\mathcal Q'}$ that obeys the conditions in the statement of the lemma: e.g., $v(\{j\})=\LQp(j)$ for every $j\in \LQp$ and $v(\{j\})=\eQp$ for every $j\notin \LQp$. Define $\LQ=\LQp, \eQ=\eQp$ and the set of linear equations $\AQ$ to be $\AQp$ with the additional constraint that $v(S\setminus \LQp)=\eQp \cdot |S\setminus \LQp|$.


If the additional query is a demand query, the adversary answers it as follows. Find a solution for the system $\AQp \cdot v = \eQp\cdot \AQp\cdot \vec{\textbf{1}}$ and $v\geq 0$, with the smallest support size. 
Since the vector $\eQp\cdot \vec{\textbf{1}}$ is a solution, by Lemma~\ref{lem:sparse} there exists a solution $\vec r\in\R_{+}^{m-|\LQp|}$ with at most $t_{i+1}{=t_i+x'}$ non-zero coordinates (recall that $i$ is the number of demand queries in $\mathcal Q'$).
According to Claim~\ref{pro:continuance}, the vector $\vec {r'}=(1-\frac{\eQ}{\eQp}) \vec r+\eQ\cdot \vec{\textbf{1}}$ is also a valid solution. 

We now show that the induction hypothesis holds for $\mathcal Q$ with any $\eQ<\eQp$. Let $\LQ$ be $\LQp$ with all items that their value is non-zero in $\vec{r}$. We set the common value $\LQ(j)$ of any such newly added item $j\in\LQ\setminus\LQp$ to be its value in the solution $\vec {r'}$. We define $\mathcal V^{\mathcal Q}$ to be all valuations in $\mathcal V^{\mathcal Q'}$ that have $v(j)=\LQ(j)$ for every $j\in\LQ\setminus\LQp$.


The set of linear equations $\AQ$ is obtained from $\AQp$ by updating every constraint $\sum_{a\in S\setminus \LQp}{v(a)}= b$ in $\AQp$ to $\sum_{a\in S\setminus \LQ}v(a)= b-\sum_{a\in S\cap \LQ}v(a) = \eQ\cdot |S\setminus \LQ|$.

Note that $\mathcal V^{\mathcal Q}$ is not empty as it contains the valuation defined by $\vec {r'}$ and the values of the items in $\LQp$. Furthermore, all valuations in $\mathcal V^{\mathcal Q}$ are consistent with all queries made so far, and in particular with the last demand query. 
{We now set $\eQ$ to be small enough to make sure that only items in $\LQ$ could be in the demand set. If $p_{min}$ is the minimal price of an item in the $(i+1)$'th (last) demand query, $\eQ$ is chosen to be strictly smaller than $\frac{p_{min}}{m}$. Since the first line in $\AQ$ has that $\sum_{a\in\items\setminus\LQ}v(a)=\eQ\cdot|\items\setminus\LQ|$, this guarantees that for any item $j\notin\LQ$, the maximal possible value for $j$ is smaller than $m\cdot\eQ<p_{min}$ and thus $j$ is not in any most profitable set for the $i+1$'th demand query. Hence, committing on the values of items in $\LQ$ is sufficient for the   implementation of the demand query.}
The adversary then answers the demand query according to the valuation defined by the solution $\vec {r'}$. 
\end{proof}

Using Claim \ref{claim-invariant} we are now ready to complete the proof of the theorem.
Let $q_v$ be the total number of value queries and $q_d$ be the total number of demand queries made by the algorithm. Let $\mathcal Q$ be the list of queries.
Since the $i$'th 
demand query adds at most $t_i$ new items to the set $\LQ$ and $t_i$ is at most $q_v$, the total number of items in $\LQ$ is $q_v\cdot q_d$. Thus, the linear system has at most $q_v$ constraints. 
Furthermore, each item $j\in \LQ$ corresponds to a linear constraint of the form $v(j)=t$ for some known value $t$. That is, the set $\mathcal V^{\mathcal Q}$
is non-empty and defined by at most $q_v\cdot q_d+q_v$ linear constraints. When $q_d,q_v \leq \sqrt{m}-1$, we have that the total number of constraints is at most $q_v\cdot q_d+q_v<m-1$. By Lemma~\ref{lem:n_value_queries} the algorithm fails for some valuation $v$.
\end{proof}

\subsection{Impossibilities for Additive Valuations using Value Queries}\label{sec:value_lb}
For deterministic algorithms, we show that $m-1$ value queries are needed to find an \opt{(m-1)} set. 
For randomized algorithms, we show two different proofs that $\Om{\frac{m}{\log m}}$ queries are needed to find a \opt{k} set (with each proof using a different $k$) with non-negligible success probability.

\begin{proposition}
    Let $A$ be a deterministic algorithm that given an additive valuation $v$ makes $q$ value queries and returns 
    an \opt{(m-1)} set
    (equivalently, finds the item with the smallest value). Then, the algorithm makes at least 
    {$m-1$} 
    value queries. 
\end{proposition}
\begin{proof}
We construct a hard valuation $v$ ``on the fly'' by observing the queries that $A$ makes and (partially) defining $v$ appropriately ($A$ is deterministic so the queries that it makes are only a function of the values that were returned in the previous queries). We will see that for every set of queries that $A$ makes, after $m-2$ queries there are at least two ways to complete the definition of $v$, and each completion has a different minimal-value item. Thus, $A$ cannot determine the minimal item after $m-2$ queries.

Specifically, in each of the first $m-2$ queries that $A$ makes, we partially define $v$ as follows: whenever the algorithm queries $v(S)$, set $v(S)=|S|$. Note that this construction is indeed valid in the sense that it is easy to extend it to a fully defined additive valuation (e.g., by setting $v(S)=|S|$ for every bundle $S$). 

The queries and answers naturally define a set of $m-2$ linear equations with $m$ variables, each equation has the form $\Sigma_{j\in S}x_j=|S|$, for some $S$. If there are less than $m-2$ independent rows in the corresponding matrix, we add some independent rows of the same form. Note that the newly added rows correspond to additional queries that the algorithm makes. 

Since there are $m$ equations and $m-2$ variables, the matrix has infinitely many solutions. In particular, there are $m-2$ variables such that if we set each of those variable 
to $1$, there are still two free variables, denoted $w,z$. Note that $w,z$ are constrained by at most one linear equation, and since $x_i=1$ for every $i$ is a valid solution, this equation must take the form $w+z=2$. Thus, both $w=\frac 1 2, z=2-\frac 1 2$, and $w=2-\frac 1 2, z=\frac 1 2$ complete setting every other variable to $1$ to two different valid solutions to the set of equations. Moreover, each of these solutions defines an additive valuation with a different minimal-value item. 
Therefore, $A$ cannot distinguish between these two valuations and find the minimal-value item if it makes less than {$m-1$} 
queries. 
\end{proof}

\begin{proposition}\label{pre:random_lb_additive_information} 
    Let $A$ be a randomized algorithm that given an additive valuation $v$ makes $q$ value queries, and with a constant positive probability returns a \opt{(m/2)} set. 
    Then, the algorithm makes in expectation at least $\Om{\frac{m}{\log m}}$ value queries.
\end{proposition}
\begin{proof}
Consider the following distribution over additive valuations: choose at random set of $m/2$ items and set the value of each item in the set to $1$. Set the value of the rest of the items to $0$. By Yao's principle, it is enough to show that every deterministic algorithm that succeeds with probability $\frac 2 3$ makes at least $\Om{\frac{m}{\log m}}$ value queries in expectation.

Fixing such a deterministic algorithm, there is a family of $\frac 2 3 \cdot {\binom{m}{\frac{m}{2}}}$ additive valuations on which the algorithm succeeds. We associate each such valuation  with a transcript of the run of the algorithm on this valuation. Such transcript consists of $q$ answers to value queries, each takes $\log \frac m 2$ bits to write down since the possible values are integers between $0$ to $\frac m 2$. Since the algorithm is deterministic we have that the identity of the bundle $S$ that is queried is only a function of the answers to the previous queries. Thus representing each transcript takes $q\cdot \log \frac m 2$ bits. Note that each such transcript must be different for every valuation that the algorithm succeeds on, and thus we have that $\frac 2 3 \cdot\log\binom{m}{m/2} \leq  q\cdot\log(m/2)$, and since $\binom{m}{m/2}>\frac{2^{m}}{m+1}$ we have that $q\geq \frac{2}{3} \cdot\frac{m-\log(m+1)}{\log(m/2)}=\Om{\frac{m}{\log m}}$ as needed.
\end{proof}

\begin{proposition}\label{pre:random_lb_additive_DISJ} 
    Let $A$ be a randomized algorithm that given an additive valuation $v$ makes $q$ value queries and returns a maximal value item (\opt{1} set) with a constant positive probability. Then, the algorithm makes in expectation at least $ \Om{\frac{m}{\log m}}$ value queries.
\end{proposition}
\begin{proof}
We prove our lower bound using a reduction to the Set Disjointness problem.
In this problem, Alice and Bob, are given two input vectors $x=(x_1,...,x_m),~y=(y_1,...y_m)\in\set{0,1}^m$ respectively, and they wish to determine whether there exists an index $i$ such that $x_i=y_i=1$. The randomized communication complexity of DISJ is $\Om{m}$, see, e.g., \cite{razborov1990distributional}.


Given an input $x,y\in\set{0,1}^m$ for the disjointness problem, define an additive valuation by setting the value of each item $w_i$ to be $v(\{w_i\})=x_i+y_i$. We follow the run of $A$ and simulate the value queries that $A$ makes: when $A$ queries the value of the bundle $S$, Alice uses $\log m$ bits to send $\Sigma_{j\in S}x_j$ and Bob uses $\log m$ bits to send $\Sigma_{j\in S}y_j$. The value $v(S)$ is simply the sum of these two numbers. Thus we have that if $A$ makes $q$ queries then it can be simulated for this family of instances by a communication protocol that takes ${2\cdot }q\cdot \log m$ bits.

Observe that the maximal value of a single item in $v$ equals to $2$ if and only if there exists an index $i$ for which the $i$-th variable has value $1$ for both Alice and Bob. Therefore, the communication protocol must use $\Omega(m)$ bits, which immediately implies that $q= \Om {\frac m {\log m}}$.
\end{proof}

	\bibliographystyle{plainnat}
    \bibliography{oracle_complexity} 
	\appendix	
	\section{Missing proofs from Section~\ref{sec:algorithms}}\label{sec:algorithms_app}

\subsection{Proof of Theorem \ref{thm:findmaxk}}\label{sec-additive}


The following two lemmas will be useful for constructing an algorithm to find \opt{k} set for additive valuations with small number of value and demand queries. 

	\begin{lemma}\label{lem:random-larger}
	For a subadditive valuation $v$ over a set $\items$ of $m$ items and a threshold $t$. 
    Let $L$ be the set of elements of value (as a singleton) larger than $t$, that is $L= \{ k\in \items| v({k})>t\}$.

	There exists an algorithm such that for any $v$ and $t$, if $L$ is not empty returns an element from $L$ picked uniformly at random. The algorithm makes $\O{\log m}$ value and demand queries.
	
    \end{lemma}
    \begin{proof}
    Denote the set returned by a demand query with uniform price $t$ by $D$. 
    If $t\cdot |D|=v(D)$ we have that the profit from $D$ is zero ($v(D)-p(D)=0$) which implies that there is no item $d\in \items$ such that $v(d)>t$, otherwise the profit from buying item $d$ alone is positive, contradicting that $D$ is a set in demand. Thus, if $t\cdot |D|=v(D)$ then $L$ is empty and the  algorithm terminates.
    
    Otherwise, $v(D)>t\cdot |D|$. By subadditivity, there exists $d\in D$ such that $v(d)>t$ and thus $L$ is non-empty.
    Given $t$, the algorithm works as follows:
    It keeps a set $N$ containing all items who might be in $L$, initiated as $N=\items$. It then goes iteratively:
    \begin{itemize}
    \item  If $|N|=1$, the single element of $N$ is in $L$, and we return it. 
    \item Otherwise, split the set $N$ into two random sets, $N_0$ and $N_1$, of sizes as equal as possible (equal up to one item, that is, $||N_0|-|N_1||\leq 1$).
    For each one of the two sets $N_0$ and $N_1$, the algorithm makes a uniform price demand query at price $t$. Denote the returned sets by $D_0$ and $D_1$.
	If the returned set $D_i$ is empty or satisfies $t\cdot |D_i|=v(D_i)$, as before, the set contains no element from $L$, and is discarded.
	As $v(D)>t\cdot |D|$, by subadditivity, $v(D_i)>t\cdot |D_i|$ for at least one $i\in\{0,1\}$, so at least one of them is not discarded.
	Thus, $D_i$ is discarded if and only if it does not intersect $L$. 
	We pick one of these non-discarded sets at random. Denote it corresponding set $N_i$ by $N$, and recursively run the algorithm on this new set $N$ and threshold $t$. 
	\end{itemize}

	First, observe that the algorithm terminates in  $\O{\log m}$ demand queries, as the size of $D$ shrinks by a constant factor (about $2$) at every iteration. 
	
	Finally, we observe that if $L$ is non-empty then every element in $L$ has the same probability of being picked. First observe that only elements in $L$ are ever returned. 
	Secondly, at each iteration, one set $N_i$ is discarded. It can be since $v(D_i)=t\cdot|D_i|$, in that case no item from $L$ is discarded.
	In the other case, both $N_0$ and $N_1$ contains items from $L$ and the probability any item from $L$ to be discarded in this round is equal.
    As this claim is true for every $j\in L$, all elements in $L$ have the same probability of being selected. 
    \end{proof}

If the valuation is additive, we can use Lemma \ref{lem:random-larger} to sample a random element \emph{smaller} than a given threshold in $\O{\log m}$ demand queries.
\begin{lemma}
\label{lem:random-smaller-addiitive}
    For an additive valuation $v$ over a set $\items$ of $m$ items and threshold $t$, let $Z$ be the set of elements of value (as a singleton) smaller than $t$, that is $Z= \{ k\in \items| v({k})<t\}$.
	
    There exists a randomized algorithm such that for any $v$ and $t$, if $Z$ is not empty returns an element from $Z$ picked uniformly at random, using $\O{\log m}$ value and demand queries in expectation.
	
\end{lemma}
\begin{proof}
 Let $W$ be a large enough constant, say $W=v(\items)+1$.
 Consider the additive valuation $v'$ in which for every item $j\in \items$ the value of $j$ is $v'_j=W-v_j$.
 Picking a random item from $Z= \{ k\in \items| v({k})<t\}$ is equivalent to picking a random item from 
 $L'= \{ k\in \items| v'({k})> W-t\}$.
 We now use Lemma \ref{lem:random-larger} with valuation $v'$ and threshold $W-t$ to sample a random item from $L'$ (or equivalently, a random element form $Z$). 
 To use the lemma we observe that both value and demand queries for $v'$ can be simulated by value and demand queries on $v$. 
 For value query on a set $S$ the value $v'(S)$ is simply $v'(S) = |S|\cdot W-v(S)$, and a demand query on $v'$ with price $p_j$ for item $j$ is replaced by a demand query on $v$ with price $W-p_j$.
\end{proof}


We next use the above lemma to find, for any fixed value $t$, all items of value exactly $t$. 
\begin{lemma}\label{lem:findall}
There exists a randomized algorithm that for any additive valuation $v$ over a set $\items$ of size $m$ and a value $t\in \R$, using value and demand queries finds all items in $\items$ with value exactly $t$. The algorithm makes in expectation $\O{\log^2 m}$ queries.
\end{lemma}
\begin{proof}
	We start with a uniform demand query where the price of each item is $t$, splitting the items into a demand set $D$ and its complement $C=\items\setminus D$.
	Any item of value exactly $t$ can belong to any of the two sets. 
	We show how to find all items of value exactly $t$ in $D$, the algorithm for $C$ is the same up to trivial adjustments. 
	
	First, observe that either the set $D$ does not contain items of value $t$ or that $t$ is the minimal value in the set.
	A uniform demand query for a price that is strictly larger than $t$ but strictly smaller than any other item in $D$ will return all items of value larger than $t$ in $D$, and only them.	
	In order to find such a price, we run the following iterative algorithm:
	We maintain a set $S$, initiated to have all items in $D$. 
	The set will contain any item in $D$ that might have value $t$.
	If the average value of an item in $S$ is $t$ ($v(S)=t\cdot |S|$), we are done (all elements in $S$ must have value $t$). Otherwise, there is at least one element of value higher than $t$.   
	
	We now use Lemma \ref{lem:random-larger} on the set $S$ as the set of all items and the threshold $t$ to pick an element from $L= \{ k\in S| v({k})>t\}$  uniformly at random, in $\O{\log m}$ demand queries. Let $q$ be the value of the element picked.

	We have that in expectation, $q$ is smaller or equal to at least half of the items in $L$. 
	Let $p=\frac{q+t}{2}$ and note that $q>p>t$. 
	A uniform demand query on $S$ at price $p$ returns a set of items, each of value larger than $t$, and that set is of size at least half the size of $L$, in expectation (as it includes all items of value at least $q$). 
	We remove all these demanded items from $S$, and iterate. As the problem size shrinks by factor of two, $\O{\log m}$ rounds, each with $\O{\log m}$ 
	demand queries, 
	suffice in expectation in order of identifying all the set $L$ and remove it from $S$ completely, leaving in $S$ exactly the set of items from $D$ that have value exactly $t$.
	
	The set $C$ is handled in a similar way with some minor changes.
	First, when using Lemma \ref{lem:random-smaller-addiitive} we do so to pick an element from $Z= \{ k\in S| v({k})<t\}$  uniformly at random (instead of picking from $L$) and denote its value by $v$. 
	Unlike for $D$, to keep items of value $t$ at $S$ we now update $S$ to include items that are \textbf{in} the demand for uniform price $p=\frac{v+t}{2}$ which now satisfies $t>p>q$. 
\end{proof}

We can now conclude the proof of Theorem \ref{thm:findmaxk}.

    The algorithm gradually builds a \opt{k} set $K$ by maintaining a set $S$ that contains items that are still candidates for inclusion in the set $K$. It runs in iterations, in each iteration items are either moved from $S$  to $K$ (if we determine that they are among the $k$ items with the highest value), or removed from $S$ (if we know for sure they are not). In expectation, $S$ will shrink by a constant factor at each round, so the expected number of rounds is $\O{\log m}$ rounds. 
    The algorithm continues until $K$ is an \opt{k} set. We move to present the algorithm more formally. 


	
	Given the parameter $k$, the algorithm works as follows. It first initializes $S=\items$ and $K=\emptyset$. Then it runs the following iterative procedure:
	\begin{itemize}
	    \item 
	    Sample a random item $t\in S$. Make a uniform demand query with its value as the price for each item. In addition, using Lemma~\ref{lem:findall}, find all items of value $t$, i.e., find a partition of $S$ into three sets: $S_H = \set{j\in S\mid v_j>t},~S_M = \set{j\in S\mid v_j=t}$ and $S_L = \set{j\in S\mid v_j<t}$.
	\begin{itemize}
	\item If $k-|K|<|S_H|$, we update $S$ to be the set $S_H$ (no new items are added to $K$). Reiterate with the updated $S$.
	\item If $|S_H|\leq k-|K|\leq |S_H\cup S_M|$, the algorithm adds $S_H$ to $K$ and additionally adds  arbitrary items from $S_M$ to $K$ to complete filling it with $k$ items altogether. It then returns $K$ and terminate.
	\item If $|S_H\cup S_M|< k-|K|$ then it adds $S_H\cup S_M$ to $K$, update $S$ to be the set $S_L$. Reiterate with the updated $S$.
	\end{itemize}
	\end{itemize}
	We first argue that the algorithm indeed returns a \opt{k} set. Since the valuation function is additive, for any $r<k$, a \opt{r} set can be extended to a \opt{k} set by adding items or largest value that are not in the set. The algorithm does this till $k$ items  are added. 
	
	Each iteration of the algorithm makes $\O{\log^2 m}$ demand queries in expectation. 
	In expectation, at each round, at least half of the items in $S$ are classified and the algorithm terminates after $\O{\log m}$ rounds in expectation. 

\subsection{Proof of Theorem \ref{thm:findmaxkMR}}\label{sec-mr}




We first find a maximum weight independent set $R$, and denote its rank by $r$. 
When $k\leq r$, finding a \opt{k} set reduces to finding a \opt{k} set in $R$. Since the restriction of the valuation $v$ to the set of items $R$ is additive (since $R$ is a base), we can use the algorithm described in Theorem~\ref{thm:findmaxk} to find a $k$-optimal set.
If $k>r$ then  \opt{k} set can be constructed by adding $k-r$ arbitrary items to $R$.  

We next show that a maximum weight independent set $R$ can be found using $\O{\log^2 m}$ demand queries in expectation. 
As by Theorem~\ref{thm:findmaxk} we can find a \opt{k} set in $R$ using $\O{\log^3 m}$ demand queries in expectation, the bound on the number of queries follows.

In the next proof we use the notation $v(\cdot |R)$ to denote the marginal valuation given a set $R$. I.e., for a set $S$ we have that $v(S|R)=v(S\cup R)-v(R)$. Given query oracle for $v$, both value and demand queries can be easily implemented for $v(\cdot |R)$: value queries by querying $v(S\cup R)$ and $v(R)$, and demand queries by setting zero prices for all items in $R$, which by monotonicity guarantees that all items in $R$ are in the demand.

\begin{lemma}\label{lem:findMWB}
	There exists a randomized algorithm that for any  weighted matroid-rank valuation $v$ over $m$ items, finds a maximum weight independent set $R$ using value and demand queries. The algorithm
	makes in expectation $\O{\log^2 m}$ queries. 
\end{lemma}
\begin{proof}
\cout{	First, the algorithm makes the value query $v(\items)$ for the value of the set of all items. This is also the value of any maximum weight independent set.
	Let $t = \min_{a\in \items,~\omega(a)>0}{\omega(a)}$ the value of the smallest positive item.\mbc{I think we have a result that we can find such an item for sub-additive in $\O{\log^2 m}$ queries in expectation. Can't we just use it?}
	A uniform demand query for price that is any positive value strictly smaller than $t$, returns a maximum weight independent set and no other additional item. 
	That is, our goal is to find a positive 
	number smaller than $t$.
	}
	The algorithm is iterative. First, initialize $R$ to be the empty set. At each round we consider the valuation $v(\cdot|R)$ and observe that it is subadditive. Thus, we can use Lemma \ref{lem:random-larger} to pick a uniform random item $r$ from the set of items satisfying $v(r|R)>0$, using $\O{\log m}$ demand queries in expectation.
	We then make a demand query with price $p=v(r|R)/2$ for any item in $\items\setminus R$, and zero for all items in $R$. The returned set $D$ contains $R$ and is a subset of a maximum weight independent set. 
	We update $R$ to be $D$ (which is a superset of the prior $R$), and if for the updated $R$ it holds that $v(R)<v(\items)$ we reiterate, again picking a random $r$ and so on. 
    The process ends when $R$ satisfies $v(R)=v(\items)$, and thus is a maximum weight independent set.
 
    
    Clearly, if the algorithm terminates with $R$ satisfying $v(R)=v(\items)$ then $R$ is indeed a maximum weight independent set - in every iteration $R$ is an independent set, and as $v(R)=v(\items)$ it has maximal weight.
    The algorithm must terminate with $R$ satisfying $v(R)=v(\items)$ as $R$ size monotonically increases and is bounded, and unless $v(R)=v(\items)$ there is always an item not in $R$ with positive marginal, so some item can be picked at the next iteration.  

    Finally, we claim that the expected number of demand queries the algorithm makes is $\O{\log^2 m}$. 
    Since $r$ is selected randomly from all items with positive marginal value  relative to $R$, we have that, in expectation, at least half of the items in $\items\setminus R$ have marginal of at least $p$ relative to $R$. 
    Each such item is either in the returned set $D$ (and hence in the updated $R$) or has a negative marginal utility when added to $D$. 
    That is, at each round, in expectation,  at least half of the items are either in $R$ or have a marginal 0 when added to it.
    Hence, after $\O{\log m}$ rounds in expectation, all items are either in $R$ or have a marginal 0 relative to $R$. 
    When that is the situation, $R$ is a maximum weight independent set. 
    Since each round requires $\O{\log m}$ queries, the total number of queries is still $\O{\log^2 m}$.
    \end{proof}

\section{Proof of Theorem \ref{thm:menu_expected}}
\label{app:imp}
Given a menu $\menu$, we normalize\footnote{It is trivial to adapt the queries to this normalization, e.g., divide the price of any demand query with this normalization factor.} the minimal positive price to $1$ and denote the highest price in the menu by $H_{\menu}$.


We next use Lemma 9.1 of \cite{hart2019selling} to   
show that for any bundle-size pricing menu $\menu$ there exists a bundle-size pricing menu with bundle-size pricing menu size $\O{\epsilon^{-2}\ln(H_{\menu}})$ that obtains at least $1-\epsilon$ fraction of $\menu$'s revenue.
Our lemma  makes simple observations regarding the result of \cite{hart2019selling} when applied to a bundle-size pricing menu. 
\begin{lemma}\label{lem:loglevels}
	Given a bundle-size pricing menu $\menu_1$, 
	{for any $\epsilon>0$} there exists a bundle-size pricing menu $\menu_2$ that offers at most $2+\frac{5}{\epsilon^2}\cdot \ln H_{\menu_1}$ different bundles' sizes, such that for any monotone valuation $v$, $\REV_{\menu_2}(v)\geq(1-\epsilon)\REV_{\menu_1}(v)$. Moreover, if $H_{\menu_1}>2$, the ratio of two different prices in $\menu_2$ is at least $e^{\frac{\epsilon^2}{7}}$. 
\end{lemma}
\begin{proof}
    For $\epsilon\geq 1$ the claim is trivially true with $\menu_2$ being the menu that only offers nothing for zero payment. 
    We now assume that  $\epsilon<1$.
	Denote $H=H_{\menu_1}$. 
	We start by splitting the range $[1,H]$ into $K$ subranges, each with a ratio of $H^{1/K}$ between its endpoints
	where $K$ is the smallest integer such that $H^{1/K}\leq 1+\frac{\epsilon^2}{4}$, i.e., $K \leq 1+\ln^{-1}(1+\frac{\epsilon^2}{4} )\cdot \log H$, which is at most $1+\frac{5}{\epsilon^2}\cdot \ln H$ for $\epsilon<1$. 
	All prices in the same subranges are rounded to the same single price in $\menu_2$. By the monotonicity assumption, if two different bundle sizes are offered for the same price, a larger size bundle will always be selected. Hence, for all price in $\menu_2$ we keep only entries offering the largest bundle sizes for that price and $\menu_2$ will have at most $K+1$ different price levels.

	We now explain the rounding schema.  
	For any price $s$ in $\menu_1$, we apply the transform $\phi(s)$ by rounding $s$ up to the {next multiple of $H^{1/K}$} and then multiplying it by $1-\frac{\epsilon}{2}$.
	Hence we have that, $(1-\frac{\epsilon}{2})s\leq\phi(s)$ and $\phi(s)<s\cdot {(1-\frac{\epsilon}{2})} \cdot H^{1/K}<(1-\frac{\epsilon}{2})(1+\frac{\epsilon^2}{4})\cdot s$ for any price $s$.
	We then have that for any two prices $s,s'$:
	\begin{equation}\label{eq:phi_rounding}
	\phi(s)-\phi(s') \leq 
	\left(1-\frac{\epsilon}{2}\right)\left(1+\frac{\epsilon^2}{4}\right)s - \left(1-\frac{\epsilon}{2}\right)s'= 
	s -s' - \left(\frac{\epsilon}{2}\right)\left(\left(1-\frac{\epsilon}{2} + \frac{\epsilon^2}{4} \right)s-s'\right) < 
	s-s',
	\end{equation}
    When the last inequality holds whenever $\left(1-\frac{\epsilon}{2} + \frac{\epsilon^2}{4} \right)s-s'>0$. 
    
    
	To complete the proof we show that for any valuation $v$, if the profit-maximizing set in $\menu_1$ was $A$ and has generated revenue of $s$, then  the profit-maximizing set $A'$ in $\menu_2$ generates revenue of at least $s\cdot (1-\epsilon)$. 
	Denote by $s'$ the price of $A'$ in $\menu_1$.
	Since $A$ is selected in $\menu_1$, we have that $v(A)-s\geq v(A')-s'$.
	A necessary condition for a buyer to  select $A'$ in $\menu_2$ is that $v(A)-\phi(s)\leq v(A')-\phi(s')$.
	Combining the two we get that a necessary condition  to  select $A'$ in $\menu_2$ is that $\phi(s)-\phi(s') \geq s-s'$.
	For this not to contradict  Equation~\eqref{eq:phi_rounding} it must holds that $\left(1-\frac{\epsilon}{2} + \frac{\epsilon^2}{4} \right)s\leq s'$.
	From this inequality we derive that for $A'$ that is picked in $\menu_2$ the payment is 
	$\phi(s')>(1-\frac{\epsilon}{2})s'>\left(1-\frac{\epsilon}{2} + \frac{\epsilon^2}{4} \right)(1-\frac{\epsilon}{2})s>(1-\epsilon)s$, and we conclude that  $\REV_{\menu_2}(v)\geq(1-\epsilon)\REV_{\menu_1}(v)$. 
 In addition, we have that for $H>2$ and $0<\epsilon<1$, the parameter  $K$ is smaller than $\frac{7}{\epsilon^2}\ln H$ which implies that
	$H^{1/K}>H^{(\frac{7}{\epsilon^2}\ln H)^{-1}}=e^{\frac{\epsilon^2}{7}}$ and this is the minimal ratio between two different prices in $\menu_2$.
\end{proof}

Recall that by Theorem \ref{thm:findmaxkMR}, for any weighted matroid-rank valuation $v$ over a set of size $m$, there exists a randomized algorithm that finds a \opt{k} set and has in expectation makes $\O{\log^3 m}$ demand queries. Combining this with the lemma above we get as an immediate corollary that for any bundle-size pricing menu $\menu_1$, 
there exists a bundle-size pricing menu $\menu_2$ 
such that for any weighted matroid-rank  valuation $v$ it holds that $\REV_{\menu_2}(v)\geq(1-\epsilon)\REV_{\menu_1}(v)$, and such that the buyer can find a profit-maximizing set with  $\O{\log^3 m\cdot \epsilon^{-2}\log H_{\menu}}$ demand queries in expectation.

The number of demand queries in above result depends on the price-ratio $H_{\menu}$ being not too large. 
Our main result in this section is that we can get rid of the dependence on $H_{\menu}$ when optimizing the \emph{expected} revenue for a given distribution over valuations (rather than ex-post, for any given valuation).  

Before proving our main result, we prove a lemma showing that there are cases in which many entries of a menu can be removed without harming the {\em expected} revenue by much.
\begin{lemma}\label{lem:delete_small}
    Given a bundle-size pricing menu $\menu_1$ over $m$ items, a distribution $\D$ over valuations, and $\epsilon>0$, let $d$ be smallest bundle size in $\menu_1$ such that the expected revenue from selling bundles of size $d$ is at least $\frac{\epsilon}{m}\REV(\menu_1,\D)$.
    Let $\menu_2$ be the menu obtained from $\menu_1$ by removing all entries of bundles smaller then $d$. Then, $\REV(\menu_2,\D)\geq (1-\epsilon)\REV(\menu_1,\D)$.
\end{lemma}
\begin{proof}
    For any valuation in which a bundle of size $d'\geq d$ was selected in $\menu_1$, it was picked although the buyer had the option to buy a smaller and cheaper bundle. 
    Removing an option that was not picked  will not change the selection and  
    hence, even if every cheaper bundle is removed, bundle $d'$ will also selected in $\menu_2$.
    That is, removing bundles of size $k$ smaller than $d$ from the menu will result in revenue loss that is bounded by the revenue contribution of valuations that selected those bundles of size $k$ in $\menu_1$, and that loss is at most $\frac{\epsilon}{m}\REV(\menu_1,\D)$.
    Since at most $m$ bundles are removed, $\REV(\menu_2,\D)\geq (1-\epsilon)\REV(\menu_1,\D)$.
\end{proof}

{We now prove our main result in this section by using the above two lemmas and applying our main upper bound of Theorem~\ref{thm:findmaxkMR} (which shows that 
for weighted matroid-rank valuations
we can  find a \opt{k} set in poly-logarithmic number of demand queries).}

We are now ready to 
prove Theorem \ref{thm:menu_expected}.
\cout{
\begin{theorem}
	Given a distribution $\D$ over weighted matroid-rank valuations over a set $\items$ of $m$ items, a bundle-size pricing menu $\menu_1$, and $\epsilon>0$, there exists a bundle-size pricing menu $\menu_2$ such that $\REV(\menu_2,\D)\geq(1-\epsilon)\REV(\menu_1,\D)$ and such that a profit-maximizing bundle for $\menu_2$ can be found in $\O{\epsilon^{-3}\log^4 m}$ demand queries in expectation,  where the expectation is over the internal random coins of the algorithm and the distribution $\D$. 
\end{theorem}
}

\begin{proof}[of Theorem \ref{thm:menu_expected}]
	Denote $H=H_{\menu_1}$. 
	Given $\menu_1$ we use Lemma~\ref{lem:loglevels} to construct a bundle-size pricing menu $\menu_2'$ with $\O{\epsilon^{-2}\log H}$ price levels such that $\REV(\menu_2',\D)\geq(1-\epsilon)\REV(\menu_1,\D)$.
	For $H\leq m$, we have that the number of price levels in $\menu_2'$ is $\O{\epsilon^{-2}\log m}$, and by solving for each price level separately, we can find a profit-maximizing set with  $\O{\log^4 m\cdot \epsilon^{-2}}$ demand queries in expectation. 
	
	Else, we have that $H>m$ and we now modify $\menu_2'$ using the information that valuations are drawn from $\D$.
	We use Lemma~\ref{lem:delete_small} to construct a bundle-size pricing menu $\menu_2''$, by removing entries from $\menu_2'$ that have low expected revenue contribution, and get a menu
	such that $\REV(\menu_2'',\D)\geq(1-\epsilon)\REV(\menu_2',\D)$.
	Note that since $\menu_2''$ is obtained by deleting entries, the ratio between prices is still at least $e^{7/\epsilon^{2}}$.
	%
	%
	Let $p$ the price of the cheapest bundle in $\menu_2''$. By the definition of $\menu_2''$ we have that $p\geq p\cdot\Pr[\text{cheapest bundle selected by the buyer}]\geq \frac{\epsilon}{m}\REV(\menu_2'',\D)$.
	
	Let $t=\epsilon^{-1} m^3p$.
	If $\Pr_{v\sim\D}\left[v(\items)\geq t\right]>m^{-2}$, we define $\menu_2$ to be the menu that only sells the grand bundle for price $t$, having an expected revenue which is at least $ t\cdot m^{-2}=mp/\epsilon\geq \REV(\menu_2'',\D)>(1-2\epsilon)\REV(\menu_1,\D)$.
	This menu is implementable using a single value query.

	Else $\Pr_{v\sim\D}\left[v(\items)\geq t\right]\leq m^{-2}$. In this case we define $\menu_2$ to be $\menu_2''$ and we implement it using the following algorithm:
	We start by querying the value of the grand bundle $v(\items)$. 
	We then consider two cases. 
	
	In the first case $v(\items)\geq t$. In this case we use the fact that for matroid-rank valuations the greedy algorithm finds a \opt{k} set for every $k\in [m]$: at each point selecting the item with the highest marginal value with respect to the set selected so far. 
	The greedy algorithm requires $O(m^2)$ value queries. 
	Since 
	the probability of running this algorithm is at most $m^{-2}$, and when running the number of queries is $O(m^2)$, this case contributes only a constant to the expected query complexity, where the expectation is taken over the distribution $\D$.
	
	Else, we are in the second case in which  $v(\items)\leq t =\epsilon^{-1} m^3p$. By monotonicity no bundle is sold for a price higher than $t$ in this realization of $v$.
	Since $p$ is the cheapest price in the menu and $t =\epsilon^{-1} m^3p\geq v(\items)$, the ratio between the highest and lowest prices of sold bundles is at most $\epsilon^{-1} m^3$. 
	As $H>2$ (which is the case since $H>m$), 
	by Lemma~\ref{lem:loglevels} the ratio between any two different prices in the menu is at least
	$e^{\epsilon^2/7}$. 
	Thus,  there are at most $21\epsilon^{-3}\cdot\log m$ bundle sizes in the range between $p$ and $\epsilon^{-1} m^3p$ that are needed to be considered. 
	{For any weighted matroid-rank valuations we can find a profit-maximizing set for each bundle size separately using  $\O{\log^3 m}$ demand queries in expectation, by applying our main upper bound presented in Theorem~\ref{thm:findmaxkMR}.}
	Among those bundles, we return the one with the highest profit when paying its price.	
	In total, the expected number of queries is then $\O{\epsilon^{-3}\log^4 m}$ demand queries, where the expectation is taken over the internal random coins of the algorithm and the distribution $\D$.
	
	In both of the possible menus, $\menu_2''$ and the grand bundle, the expected revenue is at least  $(1-2\epsilon)\REV(\menu_1,\D)$.
	Setting $\epsilon' = \epsilon/2$ in the theorem statement completes the proof.	
	%
\end{proof}

{
As an immediate corollary of the theorem we get that it is possible to take the optimal bundle-size pricing revenue on $\mathcal D$, and while losing at most $\epsilon$ fraction of the revenue, convert it to another bundle-size pricing menu that has a primitive complexity of $poly(\log m, \frac 1 \eps)$.}

\section{Lower Bounds}
\subsection{Submodular Valuation (Proof of  Theorem \ref{thm:submodular})} \label{app:submodular} 
In this section we prove Theorem \ref{thm:submodular}. 
Some definitions and auxiliary claims are presented in Subsection~\ref{subsec-submodular-aux-app}. We then prove 
Lemma \ref{lem:SM_value_queries} in Section \ref{subsec-submodular-value}.

\subsubsection{Definitions and Auxiliary Claims}\label{subsec-submodular-aux-app}
We first show that every valuation $v$ in the support of $\D$ is indeed submodular. We then present some definitions and prove several claims that will be helpful in the proof of the theorem.

\begin{lemma}
	Every valuation $v$ in the support of $\D$ is submodular.
\end{lemma}
\begin{proof}
	It is sufficient to show that for any $S\subseteq M$ and $a,b\in M \setminus S$
	it holds that $v(S\cup\set{a})+v(S\cup\set{b}) \geq v(S\cup\set{a,b}) +v (S)$.
	We split into cases and verify the inequality  holds for each case:
	\begin{itemize}
		\item $|S|>k+1$: $v(S\cup\set{a})+v(S\cup\set{b}) =2k = v(S\cup\set{a,b}) +v (S)$.
		\item $|S|=k+1,~S\in\Bv$: $v(S\cup\set{a})+v(S\cup\set{b}) =2k = v(S\cup\set{a,b}) +v (S)$.
		\item $|S|=k+1,~S\notin\Bv$: $v(S\cup\set{a})+v(S\cup\set{b}) =2k > 2k - \frac {3} {11} = v(S\cup\set{a,b}) +v (S)$.
		\item $|S|=k,~S=\Gv$: $v(S\cup\set{a})+v(S\cup\set{b}) =2k -\frac {6} {11}= v(S\cup\set{a,b}) +v (S)$.
		\item $|S|=k,~S\neq \Gv$: $v(S\cup\set{a})+v(S\cup\set{b}) \geq 2k-\frac 6 {11} > 2k-\frac 7 {11} = v(S\cup\set{a,b}) +v (S)$.
		\item $|S|=k-1,~S\in\Rv$: $v(S\cup\set{a})+v(S\cup\set{b}) \geq 2k -\frac {14} {11} \geq v(S\cup\set{a,b}) +v (S)$.
		\item $|S|=k-1,~S\notin\Rv$: $v(S\cup\set{a})+v(S\cup\set{b}) = 2k -\frac {14} {11} \geq v(S\cup\set{a,b}) +v (S)$.
		\item $|S|<k-1$: $v(S\cup\set{a})+v(S\cup\set{b}) \geq 2k - \frac {28} {14} \geq v(S\cup\set{a,b}) +v (S)$.
	\end{itemize}
\end{proof}

A valuation in the support of $\D$ is completely defined by the values of all sets of size $k$ and $k+1$.
We say that a set of value queries $\Q$ is in a {\em canonical form} if all queries in $\Q$ are for sets of size $k$ or $k+1$, and for every query of size $k$ all of its supersets of size $k+1$ are also in $\Q$. 
{Essentially, 
all information that a set of value queries conveys  about a valuation can also be conveyed by some set of queries that is in a canonical form and is not much larger.}
The next proposition shows that we can assume that the query set is in a canonical form at a cost of a polynomial blow-up in the number of queries:

\begin{proposition}
Let $A'$ be an algorithm that makes $t$ value queries on a valuation  in the support of $\D$. Then, there is an algorithm $A$ that simulates $A'$ while making $m^2\cdot t$ value queries on a valuation in the support of $\D$. Moreover, the set of queries that $A$ makes has a canonical form.
\end{proposition}
\begin{proof}
We show how instead of querying directly a set $S$ we can compute its value by $m^2$ value queries to bundles of size $k$ and $k+1$. We split into cases:
\begin{itemize}
    \item $|S|<k-1$ or $|S|>k+1$: the value of $S$ is known and the query is discarded. 
    \item $|S|=k+1$: the query remains the same.
    \item $|S|=k-1$: the query is replaced with value queries to all $\frac{(m-k)(m-k-1)}{2}$ sets of size $k+1$ that contain $S$. If the value of all these sets of size $k+1$ is $k$ then $S\in \Rv$ and its value is $k-1$, otherwise $S\notin \Rv$ and its value is $k- \frac {14} {11}$.
    \item $|S|=k$: the query remains the same and we additionally make value queries to all $(m-k)$ sets of size $k+1$ containing $S$. 
\end{itemize}
\end{proof}

We next present several useful definitions and notations.  
Fix some deterministic algorithm $A$ that makes only value queries and runs on valuations from $\D$. 
Fix any valuation $v$ from the support of $\D$, and let $\Qv$ denote the list of $t$ bundles that $A$ queried together with their values.
Let $\DQv$ denote the distribution over valuations that is obtained by sampling according to $\D$ a valuation that is consistent with the queries in $\Qv$. 
    Let $\Byes$ be the family of sets that includes every set $S$ such that $\Pr_{v'\sim\DQv}[S\in\Bvp]=1$.
    Similarly, let $\Bno$ be the family of sets that includes every set $S$ such that $\Pr_{v'\sim\DQv}[S\in\Bvp]=0$.
    Let $\Qvk$ be the family of sets of size $k$ that were queried in $\Qv$.
 
We now claim that assuming queries are in a canonical form, the conditional distribution for sets not queried is essentially identical to the prior.
\begin{lemma}\label{lem:B_is_independent_app}
    Fix any valuation $v$ sampled from $\D$ and assume $\Qv$ is in a canonical form. 
    It holds that:
    \begin{itemize}
        \item For any set $S$ of size $k+1$ it holds that $\Pr_{v'\sim\DQv}[S\in\Bvp]\in \{0,1,\frac 1 {m^{2}}\}$.
        \item The conditional probabilities are independent:  for any family $\mathcal{F}$ of sets of size $k+1$ it holds that
    $\Pr_{v'\sim\DQv}[\forall S\in\mathcal{F}\ ,S\notin\Bvp]= \prod_{S\in\mathcal{F}}\Pr_{v'\sim\DQv}[S\notin\Bvp]$.
    \end{itemize}
\end{lemma}
\begin{proof}
    We start with proving the first part. First, if $S\in \Bno$ or $S\in \Byes$ then by definition $S\in \Bvp$ with probability $0$ or $1$, respectively. Thus, the first part of the claim holds for all sets of size $k+1$ that were queried in $\Qv$. Next, suppose that $S$ was not queried in $\Qv$. Since $\Qv$ is in a canonical form, if $T$ is a set of size $k$ that was queried, all of its supersets of size $k+1$ were queried. Thus, as $S$ was not queried, none of its subsets was queried. 
  
   Recall that the sampling process that defines $\Bv$ picks every set of size $k+1$ to belong to $\Bv$ independently with probability $\frac 1 {m^2}$. The posterior probability given $\Qv$ is still $\frac 1 {m^2}$, as none of sets of size $k$ that were queried is a subset of $S$, and conditional on this event the value of $S$ is sampled independently with probability  $\frac 1 {m^2}$.  Thus, by the principle of deferred decisions we can think of the membership of $S$ in $\Bv$ as determined after the values $\Qv$ are given, and thus:
    $$
    \Pr_{v'\sim\DQv}[S\in\Bvp]= \Pr_{v'\sim\D}[S\in\Bvp|\Qv=\mathcal{Q}_{v'}]=
    \frac{\Pr_{v'\sim\D}[S\in\Bvp~\&~\Qv=\mathcal{Q}_{v'}]}{\Pr_{v'\sim\D}[\Qv=\mathcal{Q}_{v'}]}
    $$
    The claim follows since $\Qv$ is in canonical form and as $S\notin \Qv$ it holds that: $${\Pr_{v'\sim\D}[S\in\Bvp~\&~\Qv=\mathcal{Q}_{v'}]}=\frac{1}{m^2}\cdot {\Pr_{v'\sim\D}[\Qv=\mathcal{Q}_{v'}]}$$
%
%
%
    For the second bullet, if $\mathcal{F}\cap\Byes$ is not empty, we have that $\Pr_{v'\sim\DQv}[\forall S\in\mathcal{F}\ ,S\notin\Bvp]=0=\prod_{S\in\mathcal{F}}\Pr_{v'\sim\DQv}[S\notin\Bvp]$.
    Otherwise, since the membership of each set $S$ in $\Bvp$ is independent in $\D$ for sets of size $k+1$, and none of the sets of size $k$ in $\Qv$ is contained in any set in 
    $\mathcal{F}\setminus\Bno$, we again can similarly apply Bayes' rule and the priciple of deferred decisions and get that:
    $$
    \Pr_{v'\sim\DQv}[\forall S\in\mathcal{F}\ ,S\notin\Bvp]=
    \Pr_{v'\sim\D}[\forall S\notin\mathcal{F}\setminus\Bno\ ,S\in\Bvp]=(1-\frac{1}{m^2})^{|\mathcal{F}\setminus\Bno|}=
    \prod_{S\in\mathcal{F}}\Pr_{v'\sim\DQv}[S\notin\Bvp]
    $$
\end{proof}

\subsubsection{An Impossibility for Value Queries}\label{subsec-submodular-value}
We next prove Lemma \ref{lem:SM_value_queries}, our lower bound for value queries.
\cout{
\begin{lemma}\label{lem:SM_value_queries}
    Fix some deterministic algorithm $A$ that makes only value queries and the set of those queries is in a canonical form. Suppose that $A$ makes $t<1.9^m$ value queries on valuations that are sampled from $\D$. Then, for a large enough $m$, the probability (over $\D$) that $A$ finds a $k$-optimal set is at most $\frac {t} {{1.9}^{m}}$.
\end{lemma}
}

\begin{proof}[Proof of Lemma~\ref{lem:SM_value_queries}]
    Fix any valuation $v$ in the support of $\D$, and let $\Qv$ denote the list of $t$ bundles that $A$ queried together with their values and assume that $\Qv$ has a canonical form. 
    

    
    {
    We now analyze the conditional distribution (given $\Qv$) that a specific set $S$ (of size $k$) is \opt{k}, showing that this probability is exponentially small even after the algorithm makes its $t$ queries, as long as $t$ is not huge.
    Thus with high probability the algorithm cannot determine which bundle is the good set.} We start with showing that by ``ignoring'' queries to bundles of size $k$. Then, we will show that the effect of such queries is small.

    Let $\D^{\Byes, \Bno}$ be the distribution over valuations obtained by sampling from $\D$ a valuation that agrees on $\Byes$ and $\Bno$. 
    \begin{claim}\label{cla:values_miss_G}
        For any set $S$, $|S|=k$, $\Pr_{v'\sim \D^{\Byes, \Bno}}[S=\Gvp]<\frac{2}{\binom{m}{k}-mt}$.
    \end{claim}
    \begin{proof}
        Observe that a $S$ of size $k$ that is a subset of some set $S'\in\Byes$ cannot be $\Gvp$ for any $v'\in supp(\D^{\Byes, \Bno})$. Each set of size $k+1$ contains $k+1$ sets of size $k$. 
        Thus there are at most $(k+1)\cdot t \leq m\cdot t$  sets of size $k$ for which $\Byes$ dictates that they cannot be $\Gvp$ for any $v'\in supp(\D^{\Byes, \Bno})$.
    
        Each of the remaining sets is not contained in any set from $\Byes$ and is contained in $m-k$ other sets of size $k+1$. There are at {least}
        $\binom{m}{k}-m\cdot t$ such sets.
        Next, for every two sets $S_1$ and $S_2$ that are not contained in any set from $\Byes$, we compare the probabilities of the events $S_1=\Gvp$ and of $S_2=\Gvp$ in a valuation $v'$ sampled from $\D^{\Byes, \Bno}$.

        For a valuation $v'$ let $\mathcal F_{v'}$ be the family of sets of size $k$ who are not contained in any of the sets from $\Bvp$. Recall that all sets of size $k$ that are in $\mathcal F_{v'}$ have the same ex-ante probability of being $\Gvp$.
        
        Consider some set $S$ of size $k$. The probability that $S=\Gvp$ when the valuation $v'$ is sampled from $\D^{\Byes, \Bno}$ is $0$ whenever $\Byes$ contains a superset of $S$. Else, none of the supersets of $S$ is in $\Byes$, and by Lemma~\ref{lem:B_is_independent}, if the number of its supersets that are in $\Bno$ is $\ell$ then $\Pr_{v'\sim \D^{\Byes, \Bno}}[S=\Gvp]=\Pr_{v'\sim \D^{\Byes, \Bno}}[S\in \mathcal F_{v'}]\cdot \frac{1}{|\mathcal F_{v'}|} = \left(1-\frac 1 {m^{2}}\right)^{m-k-\ell}\cdot \frac{1}{|\mathcal F_{v'}|}$.


        For two sets $S_1$ and $S_2$ that are not contained in any set from $\Byes$, let $\ell_i$ be the number of sets  in $\Bno$ that contain $S_i$. We have that:
        
        $$\frac{\Pr_{v'\sim\D^{\Byes, \Bno}}[S_1=\Gvp]}{\Pr_{v'\sim\D^{\Byes, \Bno}}[S_2=\Gvp]}=\frac {\left(1-\frac 1 {m^{2}}\right)^{m-k-\ell_1}} {\left(1-\frac 1 {m^{2}}\right)^{m-k-\ell_2} }= \frac 1 {\left(1-\frac{1}{m^2}\right)^{(\ell_1-\ell_2)}} \leq 2
        $$
        where the last inequality uses the fact that since $0\leq \ell_1,\ell_2\leq m-k\leq m$, we have that $\ell_1-\ell_2<m$ and thus the ratio is at most $\frac 1 {\left(1-\frac{1}{m^2}\right)^{m}}$, which approaches to $1$ as $m$ goes to infinity. 
        
        
        Assume by contradiction that for some $S$ we have that $\Pr_{v'\sim \D^{\Byes, \Bno}}[S=\Gvp]>\frac{2}{\binom{m}{k}-m\cdot t}$, {as the ratio of probabilities for any two sets $S_1,S_2$ to be the \opt{k} set is bounded by $2$, the minimal probability for each of the sets of size $k$ (of positive probability) to be the \opt{k} set is at least $\frac{1}{\binom{m}{k}-m\cdot t}$}. 
        As there are at least  $\binom{m}{k}-m\cdot t$ sets with positive probability, summing the probabilities over all of them exceeds $1$, a contradiction.
    \end{proof}
    
    We have shown that when conditioning on $\Byes$ and $\Bno$ the probability of a set being the good set is low. To fully condition on $\Qv$ we also need to condition on  $\Qvk$. 
    We show that even after conditioning on $\Qvk$ the probability for finding $\Gv$ is still small.
    \begin{claim}\label{cla:k_miss_G}
        {For $|\Qvk|<1.9^m$ and a large enough $m$, if for any set $S'$, $\Pr_{v'\sim \D^{\Byes, \Bno}}[S'=\Gvp]<\frac{1}{1.95^m}$, then for any set $S$, either $S\in \Qv$ or $\Pr_{v'\sim \DQv}[S=\Gvp]<\frac{2}{1.95^m}$. 
        Moreover, $\Pr_{v'\sim\D^{\Byes, \Bno}}[S\in \Qvk]\leq\frac{2|\Qvk|}{1.95^m}$.
        }
    \end{claim}
    \begin{proof}
        We prove the claim by induction on $r = |\Qvk|$.
        For $r=0$ the claim trivially holds. 
        Assume the claim holds for any $r'<r$, we prove that it holds for $r$.

        For a set $S\in\Qvk$, let $\Q^S = \Qv\setminus\set{S}$.
        By the induction hypothesis, the probability that a set $S\in\Qv$ is $\Gvp$ for $v'$ sampled from $\D$ conditioned on the values of sets in $\Q^S$, is at most $\frac{2}{1.95^m}$.
        
        By the induction hypothesis, since $|\Q^S|=r-1$, with probability at least $1-\frac{2(r-1)}{1.95^m}$, none of the sets in $\Q^S$ is $\Gv$.
        Together, with probability at least $1-\frac{2r}{1.95^m}$, none of the sets in $\Qv$ is $\Gv$.
        
        Fix any set $T\notin \Qvk$ of size $k$. Let $C$ be the event that all sets in $\Qvk$ are not \opt{k} (which also determines their exact values). By Bayes' theorem, conditioned on $C$ we have that
        $$\Pr_{v'\sim \DQv}[T=\Gvp|C]=\Pr_{v'\sim \D^{\Byes, \Bno}}[T=\Gvp|C]=
        \frac{\Pr_{v'\sim \D^{\Byes, \Bno}}[T=\Gvp]}{\Pr_{v'\sim \D^{\Byes, \Bno}}[C]} < 
        \frac{\frac{1}{1.95^m}}{1-\frac{2r}{1.95^m}}
        <
        \frac{2}{1.95^m}$$
        
        where the last inequality holds for $r<1.9^m$ and a large enough $m$.
        
    \end{proof}

    
    For $t<1.9^m$, $k=\frac{m}{2}$ and a large enough $m$, by Claim~\ref{cla:values_miss_G}, we have that $\Pr_{v'\sim \DQv}[S=\Gvp]<\frac{2}{\binom{m}{k}-mt}<\frac{2}{1.95^m}$ and the conditions of Claim \ref{cla:k_miss_G} hold.
    
    Considering the set $S$ returned by the algorithm as a $t+1$'th query, the probability that either $\Gv\in\Qv$ or $S=\Gv$ is at most $\frac{2(t+1)}{1.95^m}$ with probability over $\D$.
    For large enough $m$, we have that $\frac{2(t+1)}{1.95^m}<\frac t {1.9^{m}}$.
    
     
\end{proof}

\end{document}